\documentclass[aps,prx,twocolumn,superscriptaddress,nofootinbib,floatfix,10pt]{revtex4-2}

\usepackage[T1]{fontenc}
\usepackage[utf8]{inputenc}
\usepackage{amsmath,amssymb,amsthm,mathtools}
\usepackage{bm}
\usepackage{braket}
\usepackage{graphicx}
\usepackage{tikz}
\usepackage[dvipsnames]{xcolor}
\usepackage[colorlinks=true,citecolor=NavyBlue,linkcolor=NavyBlue,urlcolor=NavyBlue]{hyperref}
\usepackage{orcidlink}

\newtheorem{theorem}{Theorem}
\newtheorem{proposition}{Proposition}

\newtheorem{lemma}{Lemma}

\newcommand{\papertitle}{Exact quantification of nonlocal magic}
\hypersetup{pdftitle={Exact quantification of nonlocal magic},pdfauthor={Piotr Sierant}}
\newcommand{\Tr}{\operatorname{Tr}}

\newcommand{\cM}{\mathcal M}
\newcommand{\cP}{\mathcal P}
\newcommand{\cA}{\mathcal A}
\newcommand{\cC}{\mathcal C}

\begin{document}

\title{\papertitle}

\author{Piotr Sierant~\orcidlink{0000-0001-9219-7274}}
\email{piotr.sierant@bsc.es}
\affiliation{Barcelona Supercomputing Center, Barcelona 08034, Spain}

\begin{abstract}
Magic, or nonstabilizerness, is the resource that lifts Clifford circuits to universal quantum computation and has become a standard diagnostic of many-body states.  For a state shared between two parties, however, a basic question has remained open: how much of the magic resides in the correlations between the parties rather than in their local bases?  Isolating this \textit{nonlocal magic} requires minimizing over all local bases, an optimization that has so far resisted exact solution.  Here we solve it for the stabilizer fidelity: the nonlocal magic of every pure multiqubit state is the distance of its entanglement spectrum from the closest spectrum of Bell pairs.
The same quantity governs an apparently unrelated task: a family of states universally embezzles entanglement under local operations and classical communication if and only if its nonlocal magic diverges.  
The deciding property is not the amount of entanglement but the way the entanglement spectrum spreads its weight across factor-of-two windows of rank, so that critical chains and random-singlet states, with identical logarithmic entanglement scaling, carry unbounded and vanishing nonlocal magic, respectively.
Nonlocal magic thereby becomes an operationally meaningful property of quantum correlations, directly accessible to tensor-network simulations and, through entanglement spectroscopy, to experiments.
\end{abstract}

\maketitle

\section*{Introduction}
Entanglement~\cite{amico2008entanglement, Horodecki2009Entanglement} has become the organizing principle of quantum many-body physics.  Its scaling with subsystem size separates gapped from critical phases~\cite{RevModPhys.82.277}, its spectrum exposes topological order~\cite{LiHaldane2008} and encodes universal information about quantum critical points~\cite{CalabreseCardy2004, Calabrese2008spectrum}, and it decides whether a state admits an efficient tensor-network description~\cite{Schollwock2011}. Entanglement is not, however, the only aspect of quantumness a state can have.  A complementary perspective treats the state as a computational medium and asks whether it can power a quantum computation.
From this viewpoint entanglement must be supplemented by a second resource, magic, also called nonstabilizerness~\cite{LiuWinter2022}, which quantifies how far a state lies beyond the reach of Clifford circuits and the stabilizer states they prepare.  These circuits generate extensive entanglement yet remain efficiently simulable on a classical computer~\cite{Gottesman1998heis,PhysRevA.70.052328}, and it is magic that lifts them to universal quantum computation.

Magic is the resource that fault-tolerant architectures must distill and supply~\cite{Bravyi2005magic,PhysRevA.83.032317,Veitch2014}, it sets the cost of the classical simulation methods~\cite{PhysRevLett.115.070501, Bravyi2019LowRank} and it underlies provable quantum advantages~\cite{Zhang2024Advantage}.
Stabilizer R\'enyi entropy (SRE) and other nonstabilizer monotones have made it tractable~\cite{Leone2022SRE, Haug2023Monotones, Haug2024Algorithms}. Recent work has mapped how magic behaves in many-body systems in equilibrium~\cite{White2021CFT, Hoshino2026CFT, Oliviero2022Ising, Tarabunga2023ManyBody, Tarabunga2024Critical} and out of it~\cite{Rattacaso2023Quench, Turkeshi2025Spreading, Niroula2024MagicTransition, Tirrito2025Ergodic}, and well beyond condensed matter: in lattice gauge theories~\cite{Grieninger2026StringBreaking,Chen2026SU2}, dense neutrino gases~\cite{Hite2026Neutrino}, the decays of elementary particles~\cite{Banacki2026Higgs} and the electronic structure of molecules~\cite{Sarkis2026molecules}.

Magic and entanglement offer distinct, largely independent views of a quantum system~\cite{Gu2025Separation, Fux2024Separation, Frau2024MPS, Tirrito2024Flatness}.
While entanglement is invariant under local basis changes, magic is intrinsically basis dependent: the stabilizer states are defined by the Pauli operators of a fixed set of qubit axes. Hence, a rotation as simple as a single-qubit $T$ gate creates magic, and a product of such rotated qubits carries an extensive amount of nonstabilizerness without any entanglement.
For a state shared by two parties, however, basis dependence takes a new form. Nonstabilizerness monotones, including SRE, cannot distinguish magic associated with the local basis of either party from that encoded in their correlations.
The product of single-qubit magic states illustrates the problem: all of its magic stems from the local basis choice, and a change of local basis maps it to a computational-basis product state with no magic at all.

To isolate the component of magic carried by the correlations between the two parties, it was proposed to minimize magic over all local unitaries~\cite{Cao2025magical, Qian2025Nonlocal}. 
The resulting \emph{nonlocal magic} is the nonstabilizerness that no local change of basis can remove, and hence also the part that two separated parties cannot manufacture on their own.
Nonlocal magic has been connected to gravitational backreaction in holography~\cite{Cao2025magical}, linked to operator entanglement~\cite{Andreadakis2026Operator}, studied in quantum fields and scattering processes~\cite{Cepollaro2025Harvesting,Robin2026Scattering}, extended to fermionic Gaussian states~\cite{Iannotti2026Fermionic,Collura2026FreeFermion}, and measured at the two-qubit level on a superconducting processor~\cite{Ahmad2025Experimental}.

Removing the local basis dependence comes, however, at the price of a formidable optimization problem: the minimum runs over two local unitary groups with exponentially many continuous parameters, and it must be found for a magic measure that is itself a nontrivial function of the state.  
Despite several conjectures about its solution~\cite{Huang2026Spectral, Torre2026Spectrum, Liu2026Schmidt,
Franchini2026Schmidt}, this optimization has remained unsolved, and the nonlocal magic of a many-body state has eluded exact quantification. This leads to a natural question: \emph{Is the nonlocal magic an exactly computable property of a many-body state, and, if so, what does it measure?}

Here we answer both parts of this question for nonlocal magic, and the key turns out to be the stabilizer fidelity: optimized over local bases, it can be computed exactly, and what it measures is the shape of the entanglement spectrum.  We define the nonlocal min-relative entropy of nonstabilizerness (NMRE) by maximizing the stabilizer fidelity over all local unitaries, and we solve the nested optimization exactly for every pure multiqubit state and every, possibly unequal, bipartition. 
This is striking, because, as a magic measure, the stabilizer fidelity is among the hardest to evaluate: in a fixed basis it requires a search over all $2^{\Theta(N^2)}$ stabilizer states~\cite{PhysRevA.70.052328}, and exact algorithms currently reach only about nine qubits~\cite{Hamaguchi2024Handbook,Hamaguchi2025Faster}.
However, as a measure of nonlocal magic, stabilizer fidelity becomes simple: the optimization over local unitaries does not compound the difficulty but removes it.

The solution of the optimization problem rests on a simple fact, the \emph{quantization of stabilizer entanglement}: across any cut, a stabilizer state is, up to local Clifford operations, a definite number $k$ of Bell pairs, so its entanglement spectrum is flat with $2^k$ equal entries.  Only these \emph{Bell-pair sectors} compete. 
The ordered Schmidt coefficients enter only through their partial sums up to the \emph{dyadic} ranks $2^k=1,2,4,8,\ldots$, and a sorted computational-basis representative state attains the optimum and thereby fixes the value of NMRE (Fig.~\ref{fig:concept}).  Nonlocal magic is therefore the distance of the entanglement spectrum from the nearest Bell-pair sector, an explicit spectral quantity that can be read off from any tensor-network simulation and measured by entanglement spectroscopy. 
We  bound NMRE by the \emph{logarithm} of entanglement entropy and show that it controls the other measures of nonlocal magic: the nonlocal SRE is equivalent to it up to universal constant factors, and stabilizer extent and robustness are bounded by it. Nonlocal magic is thereby pinned down as one universal quantity, and the NMRE is its exactly solvable representative.

The exact solution then reveals that nonlocal magic controls a task that seems to have nothing to do with computation: entanglement embezzlement. Under local operations and classical communication (LOCC), entanglement can never be created, but it can be borrowed. From a suitable catalyst state, two parties can extract entanglement by LOCC while returning the catalyst almost unchanged~\cite{vanDam2003embezzling}, and a family of states is a \emph{universal embezzler} if this works for every target state, with an error that vanishes along the family. We prove that the two phenomena are governed by  one and the same spectral quantity: the largest Schmidt weight in any \emph{octave} of rank (throughout, an octave is a window of ranks differing by at most a factor of two). 
This weight decides universal  embezzlement~\cite{ZanoniTheurerGour2024}, and we prove that it fixes nonlocal magic up to a sharp universal factor.  A family of states universally embezzles entanglement if and only if its nonlocal magic diverges.  What decides both is not how much entanglement a state carries but how its Schmidt weight is spread across octaves.

The picture of many-body states this gives is unlike the one drawn by entanglement entropy. A random state, despite volume-law entanglement, packs its Schmidt weight into a few octaves and carries only a bounded amount of nonlocal magic.  A critical chain, with far less entanglement, spreads its weight over ever more octaves as it grows: its nonlocal magic increases without bound and it embezzles universally.  A random-singlet state has the same logarithmic entanglement scaling as a critical chain, yet it is nothing but a collection of Bell pairs and has no nonlocal magic at all.  While entanglement entropy cannot tell these three apart, nonlocal magic can.

  \begin{figure}[t]
   \centering
   \includegraphics[width=\columnwidth]{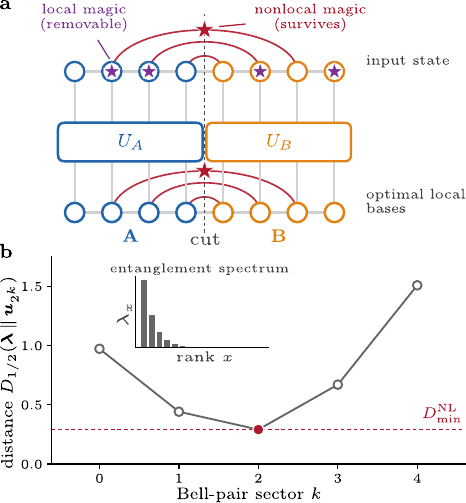}
   \caption{\textbf{Quantifying nonlocal magic.}
   \textbf{a}, A pure state $\ket{\Psi}$ of qubits shared between two
   parties $A$ and $B$ carries magic in two places: on individual qubits
   (purple stars), where a local change of basis removes it, and in the
   correlations across the cut (red star), where no local operation can.
   Maximizing the stabilizer fidelity over local unitaries
   $U_A\otimes U_B$ discards the former and defines the nonlocal
   min-relative entropy of nonstabilizerness $D_{\min}^{\rm NL}$,
   Eq.~\eqref{eq:nonlocal-fid-magic}.
   \textbf{b}, Across any cut, a stabilizer state is equivalent under
   local Clifford operations to $k$ Bell pairs, so its entanglement
   spectrum is flat with $2^k$ equal entries.  The nonlocal magic of
   $\ket{\Psi}$ is the R\'enyi-$1/2$ distance of its entanglement
   spectrum $\bm\lambda$ (inset) from the nearest of these Bell-pair
   sectors $k$.
   The minimum over $k$ (red) is $D_{\min}^{\rm NL}$.}
   \label{fig:concept}
  \end{figure}

\section*{Results}

\subsection*{Nonlocal magic from stabilizer fidelity}

Let $\cP_N:=\{I,X,Y,Z\}^{\otimes N}$ denote the $N$-qubit Pauli strings and let $\cC_N$ be the Clifford group, the unitaries that map Pauli strings to Pauli strings.  
Pure stabilizer states are the joint $+1$ eigenstates of $2^N$ commuting Pauli operators~\cite{Gottesman1997Stabilizer}, or equivalently the Clifford orbit $\cC_N\ket{0^N}:=\{C\ket{0^N}:C\in\cC_N\}$ of the computational-basis state $\ket{0^N}$.
For a pure state $\ket{\psi}$, the stabilizer fidelity $F_{\rm stab}(\ket{\psi}):=\max_{C\in\cC_N}|\bra{0^N}C^\dagger\ket{\psi}|^2$ measures how well $\ket{\psi}$ is approximated by a single stabilizer state, and $D_{\min}(\ket{\psi}):=-\log_2F_{\rm stab}(\ket{\psi})$ is its min-relative entropy of nonstabilizerness~\cite{Veitch2014,Bravyi2019LowRank, LiuWinter2022}.  
The quantity $D_{\min}$ vanishes exactly on stabilizer states, is a magic monotone, that is, nonincreasing under stabilizer operations~\cite{Veitch2014,LiuWinter2022}, is subadditive under tensor products~\cite{Bravyi2019LowRank}.

For a bipartition $A \cup B$ into $n_A$ and $n_B=N-n_A$ qubits, set $\nu:=\min(n_A,n_B)$ and define
\begin{equation}
 D_{\min}^{\rm NL}(\ket{\Psi})
 :=-\log_2\!\max_{U_A,U_B}
 F_{\rm stab}\!\left((U_A\otimes U_B) \ket{\Psi} \right).
 \label{eq:nonlocal-fid-magic}
\end{equation}
We call this quantity the nonlocal min-relative entropy of nonstabilizerness (NMRE) and $F_{\rm NL}( \ket{\Psi}):=2^{-D_{\min}^{\rm NL}(\ket{\Psi})}$ the nonlocal stabilizer fidelity.  The unitaries $U_A$ and $U_B$ are arbitrary unitaries on all qubits of $A$ and of $B$, respectively. Maximizing over them discards every local basis choice, so that what survives is magic that neither party can remove on its own (Fig.~\ref{fig:concept}a). The NMRE is by construction invariant under all local unitaries and vanishes on every product state.  Taken at face value, Eq.~\eqref{eq:nonlocal-fid-magic} is hopeless beyond a few qubits.  The unitaries range over groups of dimension $4^{n_A}$ and $4^{n_B}$, and for every choice of $U_A\otimes U_B$ the fidelity itself requires a search over the $2^{\Theta(N^2)}$ stabilizer states~\cite{PhysRevA.70.052328}.

The entanglement spectrum organizes the answer.  Set $d:=2^{\nu}$, identify each label $x\in\{0,\ldots,d-1\}$ with its binary string in $\mathbb F_2^{\nu}$, and let $\ket{\Psi}=\sum_{x=0}^{d-1}\mu_x\ket{a_x}_A\ket{b_x}_B$ with $\mu_0\ge\mu_1\ge\cdots\ge0$ be the zero-padded Schmidt decomposition, so that $\sum_x\mu_x^2=1$ and the probabilities $\bm\lambda=(\mu_x^2)_x$ form the entanglement spectrum.  Throughout, $\ket{x}_P$ denotes the $\nu$-bit encoding of $x$ on designated qubits of party $P=A,B$, tensored with $\ket{0}$ on every remaining qubit.  Our central result is that the nested optimization in Eq.~\eqref{eq:nonlocal-fid-magic} has a closed-form solution: nonlocal magic is set by the largest of $\nu+1$ fidelities, each determined by a partial sum of the ordered Schmidt coefficients.

\begin{theorem}[Exact solution]
\label{thm:main}
For every pure multiqubit state $\ket{\Psi}$ and every, possibly unequal,
bipartition,
\begin{equation}
 D_{\min}^{\rm NL}( \ket{\Psi} )
 =-\log_2\left[\,\max_{0\le k\le\nu}
 \frac{1}{2^k}
 \left(\sum_{x=0}^{2^k-1}\mu_x\right)^2\right].
 \label{eq:main-formula}
\end{equation}
Mapping the ordered Schmidt bases to the computational bases produces
the sorted computational-basis representative
$\ket{\Psi_{\rm CB}}:=\sum_{x=0}^{d-1}\mu_x\ket{x}_A\ket{x}_B$,
which globally maximizes $F_{\rm stab}$ over the local unitary orbit of $\ket{\Psi}$
and is simultaneously optimal at every allowed Schmidt rank.
\end{theorem}

The mechanism behind Theorem~\ref{thm:main} is the quantization of stabilizer entanglement.  Across any cut, a pure stabilizer state is equivalent under local Clifford operations to $k$ Bell pairs accompanied by product spectators~\cite{Fattal2004, PhysRevA.71.022315}, so its entanglement spectrum is flat, with $2^k$ equal Schmidt coefficients. 
The best stabilizer approximation of $\ket{\Psi}$ up to local unitaries is therefore the best approximation of its entanglement spectrum by a flat spectrum of dyadic rank $2^k$.
Since the number of shared Bell pairs lies between 0 and $\nu$, only $\nu+1$ Bell-pair sectors compete (Fig.~\ref{fig:concept}b), and within each sector the optimal overlap is fixed by the ordered Schmidt coefficients through von Neumann's trace inequality~\cite{vonNeumann1937,Mirsky1975}. Explicitly, the rank-$2^k$ branch
$F_k(\bm\mu):=2^{-k}\bigl(\sum_{x=0}^{2^k-1}\mu_x\bigr)^2$ is exactly the
largest fidelity between $\ket{\Psi}$ and any state that is $k$ Bell pairs up
to local unitaries.  The maximum in Eq.~\eqref{eq:main-formula} is the
upper envelope of these $\nu+1$ branches. The complete proof is given in Methods.

Remarkably, Theorem~\ref{thm:main} thus replaces two exponentially large optimizations in Eq.~\eqref{eq:nonlocal-fid-magic} by a choice among $\nu+1$ integers.  For twenty qubits, for instance, the stabilizer set contains about $10^{70}$ states~\cite{PhysRevA.70.052328} and each local unitary group has $4^{10}$ real parameters, yet the nonlocal magic across the balanced cut is the largest of eleven partial sums.  The practical consequence is that nonlocal magic becomes exactly as accessible as the entanglement spectrum itself.  Whenever the Schmidt coefficients of a state can be obtained,
$D_{\min}^{\rm NL}$ follows at the negligible cost of $\nu$ partial sums.  From a full state vector the spectrum is a single singular-value decomposition, which reaches beyond twenty qubits, whereas exact evaluation of the fixed-basis $D_{\min}$ stops at about nine qubits~\cite{Hamaguchi2024Handbook, Hamaguchi2025Faster}. For free-fermion states the spectrum follows from an $N\times N$ correlation matrix~\cite{PeschelEisler2009}, far beyond the reach of full state vector methods.  The formula is also forgiving. Only the leading $2^{k_\star}$ Schmidt coefficients enter the winning branch, and an $\ell_1$ error $\varepsilon$ in the spectrum, whether from truncation or from measurement noise, shifts $F_{\rm NL}$ by at most $2\sqrt{\varepsilon}$ (Supplementary Section~\ref{sm:properties}).

Tensor-network simulations are the natural setting for this result. A matrix-product state in canonical form carries the Schmidt coefficients across every bond explicitly~\cite{Schollwock2011}, so the nonlocal magic of a chain across every cut is read off from the bond spectra at no additional cost, in ground states obtained by the density-matrix renormalization group as well as in time-evolved states. Tree tensor networks provide the spectrum across their natural cuts in the same way.  Fixed-basis magic in the same simulations requires Pauli sampling or replica contractions~\cite{Haug2023Monotones, Tarabunga2023ManyBody}, while nonlocal magic requires nothing beyond what the simulation already stores. Because truncation errors enter only through the $\ell_1$ bound above, the bond dimension that resolves the entanglement entropy resolves the nonlocal magic as well.

Theorem~\ref{thm:main} settles also the experimental status of nonlocal magic. Measuring the SRE in a fixed basis requires Pauli expectation values or interference between copies~\cite{Haug2024Algorithms}, and isolating its nonlocal part adds an optimization over local bases that has to be run on the device.  
This is why nonlocal magic has so far been measured only for two  qubits~\cite{Ahmad2025Experimental}, where the entanglement spectrum is a single number and no optimization is needed.  For larger systems the optimization is exponentially hard, and Theorem~\ref{thm:main} removes it altogether, at any system size $N$ and any bipartition.
Because $D_{\min}^{\rm NL}$ is a functional of the entanglement spectrum alone, nonlocal magic is obtained directly by entanglement spectroscopy. Interference between copies yields the moments $\Tr\rho_A^n$ of the reduced state~\cite{Islam2015Measuring}, randomized measurements and classical shadows yield the same moments from single copies and local rotations~\cite{Brydges2019Randomized, Elben2023Toolbox}, and entanglement-Hamiltonian tomography reconstructs the entire spectrum by fitting a local entanglement Hamiltonian, as demonstrated on trapped-ion strings~\cite{Kokail2021EHTomography}.  These protocols resolve the largest eigenvalues of $\rho_A$ with the smallest statistical error, and those are precisely the ones that enter the winning branch. The $2\sqrt{\varepsilon}$ bound then converts that error directly into an error bar on nonlocal magic.
Theorem~\ref{thm:main} thus turns entanglement spectroscopy into a measurement of nonlocal magic and opens it to experiments at the many-body scale.

\subsection*{Bell-pair sectors and the geometry of nonlocal magic}

Theorem~\ref{thm:main} can be read as a statement about the geometry of entanglement spectra. Writing
$D_{1/2}(p\Vert q)=-2\log_2\sum_i\sqrt{p_iq_i}$ for the R\'enyi-$1/2$
divergence and $\bm u_r$ for the distribution that is uniform on its
first $r$ entries and zero elsewhere, Eq.~\eqref{eq:main-formula} reads
\begin{equation}
 D_{\min}^{\rm NL}(\Psi)
 =\min_{0\le k\le\nu}
 D_{1/2}\!\left(\bm\lambda\middle\Vert\bm u_{2^k}\right).
 \label{eq:renyi-geometry}
\end{equation}
Nonlocal magic is the distance of the entanglement spectrum from the finite set of flat spectra $\bm u_r$  compatible with stabilizer entanglement (Fig.~\ref{fig:concept}b). 
In particular, $D_{\min}^{\rm NL}=0$ if and only if $\bm\lambda=\bm u_{2^k}$ for some $k$, that is, if and only if the local unitary orbit of $\ket{\Psi}$ contains a stabilizer state.  
This flatness characterization of the zero set is
shared by every faithful nonlocal magic measure~\cite{Cao2025magical}.
The maximizing $k_\star$ identifies an optimal Bell-pair sector, and $F_{k_\star}$ is the best fidelity with which local unitaries can bring $\ket{\Psi}$ onto exactly $k_\star$ Bell pairs with product spectators.

At the two endpoint sectors the distance $D_{1/2}$ reduces to known entanglement measures.
The $k=0$ branch is $\lambda_0$, so $D_{1/2}(\bm\lambda\Vert\bm u_1)=S_\infty:=-\log_2\lambda_0$ is the min-entropy of entanglement, which for bipartite pure states coincides with the geometric entanglement~\cite{Shimony1995,WeiGoldbart2003}. 
The $k=\nu$ branch is the R\'enyi-$1/2$ entropy deficit $\nu-S_{1/2}$, with $S_{1/2}=2\log_2\sum_x\mu_x$.  Hence $D_{\min}^{\rm NL}\le\min\{S_\infty,\nu-S_{1/2}\}$, with equality whenever an endpoint sector is optimal: nonlocal magic never exceeds geometric entanglement.

An immediate consequence is that without entanglement there is no nonlocal magic.
A product state has $\bm\lambda=\bm u_1$ and $D_{\min}^{\rm NL}=0$, irrespective of the magic of its factors. 
For $N$ copies of the magic state $\ket{T}$ the ordinary magic measure is extensive, $D_{\min}(\ket{T}^{\otimes N}) \propto N$~\cite{Bravyi2019LowRank}, yet the NMRE of $\ket{T}^{\otimes N}$ vanishes across every cut, because local basis changes remove all magic confined to either party. 

The simplest entangled family shows how the Bell-pair sectors $k$ compete to fix $D_{\min}^{\rm NL}$. For $\ket{\psi_p}=\sqrt p\,\ket{00}+\sqrt{1-p}\,\ket{11}$ with $\tfrac12\le p\le1$, Theorem~\ref{thm:main} gives $F_{\rm NL}(\psi_p)=\max\{p,\tfrac12+\sqrt{p(1-p)}\}$: the first branch is the best product approximation ($k=0$), the second the best Bell-pair approximation ($k=1$).  The product branch wins for $p\ge p_\star=\tfrac12+\tfrac{\sqrt2}{4}$, where $D_{\min}^{\rm NL}=S_\infty$ exactly, and the Bell branch wins below; at the crossing the NMRE is continuous but has a kink.  
Analogous competition for rank-three spectra is mapped in Supplementary Fig.~\ref{fig:sector-simplex}.

The NMRE is nonnegative, continuous and invariant under local unitaries.  It is, however, \emph{not} an entanglement monotone.  Local operations convert a Bell pair deterministically into $\ket{\psi_p}$ for any $\tfrac12<p<1$~\cite{Nielsen1999Majorization} lowering the entanglement entropy while increasing the NMRE from zero to a positive value.  Nor is it additive: a flat rank-three spectrum has $F_{\rm NL}=3/4$, whereas two copies, with flat rank nine, have $F_{\rm NL}=8/9$.  These properties identify $D_{\min}^{\rm NL}$ as a geometric measure of nonlocal magic.

Two upper bounds follow from Theorem~\ref{thm:main} (Methods).  Grouping the sorted spectrum into the dyadic shells $D_0=\{0\}$ and $D_j=\{2^{j-1},\ldots,2^j-1\}$, the largest shell weight is at least $(\nu+1)^{-1}$ and never exceeds $F_{\rm NL}$, which gives the universal bound $0\le D_{\min}^{\rm NL}\le\log_2(\nu+1)$.  A complementary entropy-dependent upper bound is
\begin{equation}
 D_{\min}^{\rm NL}
 \le\log_2\!\left(2S_1+3\right),
 \qquad
 S_1=-\sum_x\lambda_x\log_2\lambda_x .
 \label{eq:entropy-ceiling-main}
\end{equation}
For many-body physics this bound has immediate implications. Entanglement scaling, the property that classifies phases of quantum matter, is converted by Eq.~\eqref{eq:entropy-ceiling-main} directly into a limit on nonlocal magic. In a fixed basis, the magic of a many-body state is generically extensive~\cite{LiuWinter2022, Haug2023Monotones, Turkeshi2025Pauli}. In contrast, nonlocal magic is bounded by the logarithm of the entanglement entropy: in terms of the smaller subsystem size $\nu$, $D_{\min}^{\rm NL}$ is at most $O(1)$ under the area law of gapped phases~\cite{RevModPhys.82.277}, at most $O(\log\log\nu)$ at a one-dimensional critical point, and at most $O(\log\nu)$ under a volume law. Thus, for every state, an optimal choice of local unitaries $U_A, U_B$ removes all but a logarithmic amount of magic, as measured by the min-relative entropy of nonstabilizerness.  The universal bound $\log_2(\nu+1)$ also shows that the regularized value of every fixed finite state vanishes, $\lim_{t\to\infty}t^{-1}D_{\min}^{\rm NL}(\Psi^{\otimes t})=0$: unlike entanglement, the NMRE has no extensive asymptotic rate.

\subsection*{Nonlocal magic and entanglement embezzlement}
\label{subsec:embezz}

The NMRE has an operational counterpart in entanglement theory.  Under  LOCC, entanglement can be spent but never created~\cite{Horodecki2009Entanglement}, and for pure states majorization of Schmidt vectors decides exactly which conversions are possible~\cite{Nielsen1999Majorization}.  Entanglement embezzlement~\cite{vanDam2003embezzling} shows how far this accounting can be stretched.  Two parties share a catalyst $\ket{\chi}$ together with unentangled ancillas $\ket{0}_A\ket{0}_B$ and apply local unitaries such that
  \begin{equation}
   (V_A\otimes V_B)\,\ket{\chi}\ket{0}_A\ket{0}_B
   \;\approx\;\ket{\chi}\ket{\sigma}.
   \label{eq:embezzlement-protocol}
  \end{equation}
The ancillas end up in the entangled target $\ket{\sigma}$ while the  catalyst is almost unchanged: the entanglement has been borrowed from $\ket{\chi}$, with the deficit hidden in an arbitrarily small distortion of the catalyst.  Embezzlement is an extreme form of catalysis~\cite{JonathanPlenio1999}, with applications to channel simulation, nonlocal games and quantum field theory~\cite{BertaChristandlRenner2011,LeungWang2014embezzling, vanLuijk2025embezzlers}.  
Equation~\eqref{eq:embezzlement-protocol} depicts the traditional communication-free local-unitary protocol. Here a sequence of catalysts $\{\ket{\chi_n}\}$ universally embezzles under LOCC if, for every fixed finite-dimensional target $\ket{\sigma}$, an LOCC channel maps the input to an output whose trace distance from $\ket{\chi_n}\ket{\sigma}$ vanishes as $n\to\infty$. 
This LOCC notion is strictly more permissive than universality under local unitaries~\cite{Luijk2025multipartite} (Supplementary Section~\ref{sm:embezzlement}).

Which spectra allow this?  Absorbing one Bell pair splits every Schmidt coefficient into two equal halves, $\lambda_x\mapsto\lambda_{\lfloor x/2\rfloor}/2$, which shifts the spectrum by one unit along the logarithmic rank axis.
For the catalyst to remain almost unchanged, this shift must leave the spectrum almost invariant, which is possible only if no single logarithmic scale carries appreciable weight.
For ordered Schmidt probabilities, padded with zeros beyond the Schmidt rank, let $\eta_{\rm emb}(\bm\lambda):=\max_{\ell\ge0}\sum_{x=\ell}^{2\ell}\lambda_x$ be the largest weight in any octave of rank $[\ell,2\ell]$.  
This heuristic is exact, as proved in Ref.~\cite{ZanoniTheurerGour2024}: $\eta_{\rm emb}(\chi_n)\to0$ is necessary and sufficient for universal LOCC embezzlement.  The following theorem ties this criterion to the NMRE.

\begin{theorem}[Embezzlement criterion]
  \label{thm:octave}
  (i) For every finite-dimensional pure state, after zero padding into
  qubit registers when necessary,
  \begin{equation}
   -\log_2\eta_{\rm emb}-\log_2(3+2\sqrt2)
   \le D_{\min}^{\rm NL}
   \le-\log_2\eta_{\rm emb},
   \label{eq:octave-equivalence-main}
  \end{equation}
  and the additive constant is asymptotically sharp.
  
  (ii) Consequently, by the criterion of
  Ref.~\cite{ZanoniTheurerGour2024}, a sequence of pure states
  universally embezzles entanglement under LOCC if and only if its
  nonlocal magic diverges:
  \begin{equation}
   \{\chi_n\}\text{ is a universal embezzler}
   \iff D_{\min}^{\rm NL}(\chi_n)\to\infty .
   \label{eq:embezzlement-iff-main}
  \end{equation}
\end{theorem}
Both bounds in part (i) come from the same picture.  The octave $[\ell,2\ell]$ that  carries the largest weight $\eta_{\rm emb}$ always fits inside a single Bell-pair sector, so that sector has fidelity at least $\eta_{\rm emb}$ with the state.  Conversely, the sector of $k$ Bell pairs is the union of the dyadic shells $D_0,\ldots,D_k$ introduced above, each of which is itself an octave and therefore weighs at most $\eta_{\rm emb}$.  Adding the shell weights with their geometric prefactors gives the constant $3+2\sqrt2$ (Methods).  Part (ii) then follows at once: $\eta_{\rm emb}\to0$ and $D_{\min}^{\rm NL}\to\infty$ are the same statement.

Theorem~\ref{thm:octave} identifies two concepts from different fields. On one side stands a resource related to universal quantum computation, the magic that survives every local change of basis and measures how far the correlations of a state lie beyond the reach of Clifford circuits.
On the other stands a capability of entanglement theory, the power of a state to act as a catalyst from which entanglement can be borrowed at will. 
The theorem says that both sides are the same property of the entanglement spectrum: a family embezzles universally exactly when its nonlocal magic diverges.

The criterion also settles a question left open above: how does nonlocal magic grow with the number of copies?  Let $\ket{\psi_{\bm p}}=\sum_i\sqrt{p_i}\ket{i}_A\ket{i}_B$ have entanglement spectrum $\bm p$, and consider the bipartite state
\begin{equation}
   \ket{\psi_{\bm p}}^{\otimes m}
   =\sum_{i_1,\ldots,i_m}\sqrt{p_{i_1}\cdots p_{i_m}}\,
   \ket{i_1\cdots i_m}_{A}\ket{i_1\cdots i_m}_{B},
   \label{eq:tensor-power-state}
\end{equation}
with the local unitaries $U_A, U_B$ acting jointly on all $A$ factors and on all $B$ factors.  The universal bound $D_{\min}^{\rm NL}\le\log_2(\nu+1)$, applied to the $m\nu$ qubits on each side, only says that $D_{\min}^{\rm NL}(\psi_{\bm p}^{\otimes m})\le\log_2(m\nu+1)$, logarithmic in $m$.  The exact behaviour, derived in Supplementary Section~\ref{sm:properties} (Theorem~S1), is a dichotomy.  If the Schmidt coefficients of $\ket{\psi_{\bm p}}$ are not all equal, $D_{\min}^{\rm NL}( \ket{\psi_{\bm p}}^{\otimes m})=\tfrac12\log_2m+O(1)$; if they are all equal, so that $\ket{\psi_{\bm p}}$ is a maximally entangled state, the nonlocal magic of any number of copies never exceeds $\tfrac12$. 
Therefore, by Theorem~\ref{thm:octave}(ii), enough copies of any partially entangled two-qubit state $\sqrt p\,\ket{00}+\sqrt{1-p}\,\ket{11}$ with $p\ne\tfrac12$ form a universal embezzler, whereas copies of a Bell pair never do. 
Combining enough copies of any bipartite state with unequal Schmidt coefficients, as in Eq.~\eqref{eq:tensor-power-state}, thus opens a generic route towards universal embezzlement beyond the canonical van Dam--Hayden catalysts~\cite{vanDam2003embezzling}.

\subsection*{Four classes of entanglement spectra}

To see what Theorems~\ref{thm:main} and~\ref{thm:octave} imply for  many-body states, we consider four representative families of entanglement spectra, which turn out to differ sharply in nonlocal magic and catalytic behaviour (Fig.~\ref{fig:classes}). 
We begin with random states, whose entanglement spectrum follows the Marchenko--Pastur law~\cite{marvcenko1967distribution, ZyczkowskiSommers2001, Nadal2011}.
\begin{figure}[t]
 \centering
 \begin{tikzpicture}
  \node[anchor=south west,inner sep=0] (img)
    {\includegraphics[width=\columnwidth]{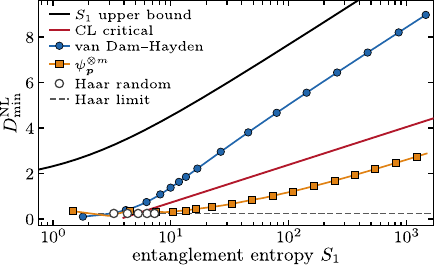}};
  \begin{scope}[x={(img.south east)},y={(img.north west)}]
   \node[anchor=west,inner sep=0] at (0.448,0.938)
     {\fontsize{7.8}{7.8}\selectfont\eqref{eq:entropy-ceiling-main}};
\node[anchor=west,inner sep=0] at (0.358,0.882)
     {\fontsize{7.8}{7.8}\selectfont\eqref{eq:CFT-main}};
   \node[anchor=west,inner sep=0] at (0.35,0.622)
     {\fontsize{7.8}{7.8}\selectfont\eqref{eq:Haar-limit-main}};
  \end{scope}
 \end{tikzpicture}
\caption{\textbf{Nonlocal magic against entanglement entropy for four
 classes of states.}  The nonlocal min-relative entropy of
 nonstabilizerness $D_{\min}^{\rm NL}$ versus the entanglement
 entropy $S_1$ on a logarithmic scale.  Black line: the entropy upper bound, which no state can exceed.  
 Red line: the Calabrese--Lefevre asymptote of critical spectra.  Blue circles: exact values for the van Dam--Hayden catalysts $\ket{\Gamma_d}$ with $d=2^2,\ldots,2^{2900}$.  
 Orange squares: exact values for the tensor powers $\psi_{\bm p}^{\otimes m}$, Eq.~\eqref{eq:tensor-power-state}, with $\bm p=(\tfrac12,\tfrac3{10},\tfrac15)$ and $m\le10^3$.  Open
 circles: random states for $N=8,\ldots,16$ qubits 
 converge to the $O(1)$ limit \eqref{eq:Haar-limit-main} for balanced cut $\Delta =0$ (dashed).  
 The families with divergent $D_{\min}^{\rm NL}$ are exactly the universal LOCC embezzlers, Eq.~\eqref{eq:embezzlement-iff-main}.}
 \label{fig:classes}
\end{figure}
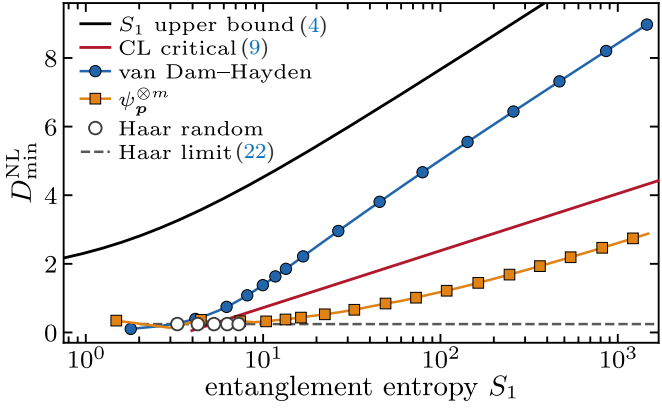

\begin{proposition}[Random states]
  \label{prop:Haar}
  Let $\ket{\Psi}$ be Haar random on $\mathbb C^d\otimes\mathbb C^D$ with $d=2^\nu\le D$, and take $d,D\to\infty$ at fixed qubit imbalance
  $\Delta=|n_A-n_B|$.  Then $D_{\min}^{\rm NL}$ converges in probability to a finite limit $D_\Delta$, given in closed form in Methods, Eq.~\eqref{eq:Haar-limit-main}, $D_0<\tfrac14$ for balanced cuts, and  $D_\Delta=2^{-\Delta}/(4\ln2),[1+O(2^{-\Delta})]$ for unbalanced ones.
\end{proposition}
Random states, the paradigm of volume-law entanglement~\cite{Page1993average}, thus carry only a fraction of a bit of nonlocal magic, as anticipated in Ref.~\cite{Cao2025magical}, and exponentially less across unbalanced ($\Delta \neq 0$) cuts.  The reason is that a random state is nearly maximally entangled across an unbalanced cut: its Marchenko--Pastur spectrum is almost uniform at the largest allowed rank, which is itself a Bell-pair sector, and $D_\Delta$ follows in closed form from the full-rank branch of Eq.~\eqref{eq:main-formula}.
Across a balanced cut ($\Delta=0$) the spectrum is far from uniform  and the sector of half the maximal rank takes over, which yields the closed form of $D_0$. Both expressions are given in Methods.  
In either  case the Schmidt weight is concentrated within a few octaves $[\ell,2\ell]$ of the maximal rank, so $\eta_{\rm emb}$ stays bounded away from zero: Haar-random families are not universal embezzlers, despite their volume-law entanglement.

Clean critical spectra provide the converse case.  Their entanglement
entropies and spectra have universal conformal-field-theory (CFT)
structure~\cite{CalabreseCardy2004,Alba2018Cardy}, captured by the
smoothed Calabrese--Lefevre (CL) spectrum~\cite{Calabrese2008spectrum},
in which $b:=-\ln\lambda_{\max}$ sets the entanglement entropy
$S_1=2b/\ln2+O(1)$ and, for a bulk interval,
$b=(c_{\rm CFT}/6)\ln L_{\rm eff}+O(1)$ in terms of the chord length
$L_{\rm eff}$.  A saddle-point evaluation of the branches of
Eq.~\eqref{eq:main-formula} (Methods and Supplementary
Section~\ref{sm:CL}) gives the following.

\begin{proposition}[Critical spectra]
  \label{prop:critical}
  For a smoothed Calabrese--Lefevre spectrum with $b\to\infty$,
  \begin{equation}
   \begin{aligned}
   D_{\min}^{\rm NL}
   &=\frac12\log_2S_1+\frac12\log_2\frac{\pi\ln2}{8}+o(1)\\
   &=\frac12\log_2\ln L_{\rm eff}
     +\frac12\log_2\frac{\pi c_{\rm CFT}}{24}+o(1).
   \end{aligned}
   \label{eq:CFT-main}
  \end{equation}
  \end{proposition}
The dominant ranks, of order $e^{2b}/\sqrt b$, grow only polynomially with the block size, far below the Schmidt rank $2^\ell$ of a block of $\ell$ spins, so the proposition applies directly to critical chains. 
The NMRE of a critical state grows as half the logarithm of its entanglement entropy, saturating the order of the bound in Eq.~\eqref{eq:entropy-ceiling-main}. The central charge enters only as the subleading constant.  By Theorem~\ref{thm:octave}, $\eta_{\rm emb}=\Theta((\ln L_{\rm eff})^{-1/2})$, so critical chains universally embezzle under LOCC.

Logarithmic entanglement alone does not imply unbounded nonlocal magic. An ideal random-singlet fixed-point state is a product of singlets, each equivalent to $\ket{\psi_{p}}$ with $p=\tfrac12$ up to local unitaries, over a random pairing of spins.  Across a cut crossed by $n_\times$ singlets it is therefore exactly $n_\times$ copies of $\ket{\psi_{1/2}}$, the maximally entangled case of Eq.~\eqref{eq:tensor-power-state}: $D_{\min}^{\rm NL}=0$, and by the dichotomy above the growing family $\ket{\psi_{ /q}}^{\otimes n_\times}$ never embezzles, despite its logarithmic disorder-averaged  entanglement~\cite{RefaelMoore2004,FagottiCalabreseMoore2011}.
Critical chains of Proposition~\ref{prop:critical} and random-singlet states thus share logarithmic entanglement scaling yet belong to opposite classes of nonlocal magic scaling: what decides is not the amount of entanglement but its delocalization across octaves $[\ell,2\ell]$.

As a final example we consider the canonical van Dam--Hayden catalysts $\ket{\Gamma_d}$~\cite{vanDam2003embezzling}, with the harmonic entanglement spectrum $\lambda_i=(iH_d)^{-1}$, $i=1,\ldots,d$, where $H_d=\sum_{i=1}^d i^{-1}$.  Theorem~\ref{thm:main} yields $D_{\min}^{\rm NL}( \ket{\Gamma_d})=\log_2\ln d-2+o(1)$ (Supplementary Section~\ref{sm:embezzlement}).  
For $d=2^\nu$ this lies within $O(1)$ of the entropy bound~\eqref{eq:entropy-ceiling-main}. Since $\eta_{\rm emb}=1/H_d$ exactly, it also lies about two bits below $-\log_2\eta_{\rm emb}$, close to the sharp constant of Theorem~\ref{thm:octave}(i).  The nearly uniform Schmidt weight of these states across logarithmic rank scales, roughly $\ln2/H_d$ per octave, keeps them far from every single Bell-pair sector.

\subsection*{$D_{\min}^{\rm NL}$ controls other nonlocal magic measures}
The NMRE is one nonlocal magic measure among several, and one may ask whether the behaviour established above, the growth laws and the embezzlement criterion, is specific to $D_{\min}^{\rm NL}$ or is shared by the other measures. The natural test is provided by the measure most widely used in many-body studies, the SRE. For a pure $N$-qubit state, let $Q_N(\psi):=2^{-N}\sum_{P\in\cP_N}\langle\psi|P|\psi\rangle^4$ be its Pauli fourth moment, so that $-\log_2Q_N$ is the second SRE~\cite{Leone2022SRE}, and let
  $
  \cM_{\rm NL}^{\rm SRE}(\ket{\Psi})
  :=-\log_2\max_{U_A,U_B}Q_N\!\left((U_A\otimes U_B)\ket{\Psi}\right)
  $
be the nonlocal SRE. Its optimization over local unitaries $U_A,U_B$ remains unsolved~\cite{Torre2026Spectrum,Liu2026Schmidt,Franchini2026Schmidt}. We show that the outcome of this optimization is nevertheless known: the NMRE fixes the nonlocal SRE up to universal constant factors.

\begin{theorem}[Equivalence of NMRE with the nonlocal SRE]
  \label{thm:SRE}
  For every bipartition, with $C_\star>1$ a universal, qubit-number-independent constant,
  \begin{equation}
   \max\left\{\frac{D_{\min}^{\rm NL}}{C_\star},\,
   \log_2\!\frac{2}{1+2^{-D_{\min}^{\rm NL}}}\right\}
   \le\cM_{\rm NL}^{\rm SRE}\le4\,D_{\min}^{\rm NL}.
   \label{eq:SRE-fidelity-main}
  \end{equation}
\end{theorem}
The two measures are therefore equivalent up to dimension-independent factors on every cut: they vanish on the same states and diverge for exactly the same families of states, so that the exactly solvable NMRE certifies the nonlocal SRE. By Theorem~\ref{thm:octave}, the nonlocal SRE of a family diverges if and only if the family is a universal entanglement embezzler.
The upper bound and the second, saturating lower bound in Eq.~\eqref{eq:SRE-fidelity-main} follow from pointwise inequalities between fidelity and Pauli
  moments~\cite{Haug2023Monotones,Haug2024Algorithms}, evaluated in optimal frames. The saturating bound behaves as $D_{\min}^{\rm NL}/2$ at small nonlocal magic but never exceeds one bit.
The linear lower bound, which takes over at large nonlocal magic, follows from the inverse theorem for the Gowers-$3$ norm of quantum states, $F_{\rm stab}\ge(Q_N)^{C_\star}$~\cite{ArunachalamBravyiDutt2024, Tarabunga2025BMSA}, proved with explicit exponents in Ref.~\cite{BaoDordrechtHelsen2025}: evaluating it in the purity-optimal local frame gives the first argument of the maximum in Eq.~\eqref{eq:SRE-fidelity-main}.
A dyadic-shell theorem for the computational-basis stabilizer purity of Ref.~\cite{Sierant2026CBOptimality}, restated in Supplementary Section~\ref{sm:sre}, bounds the optimal constant on balanced cuts: Eq.~\eqref{eq:SRE-fidelity-main} holds with $C_\star=5$.

The stabilizer fidelity also certifies the max-relative measures, which sit at the opposite end of the resource-theoretic hierarchy from the min-relative entropy.  
The stabilizer extent $\xi$~\cite{Bravyi2019LowRank}, which sets the cost of the classical simulation algorithms for magic states, and the generalized robustness $R_{\rm g}$~\cite{LiuWinter2022}, which measures how much stabilizer noise a state tolerates, both bound the NMRE from above in a fixed basis.  Optimized over local bases, the inequalities survive: $\min_{U_A,U_B}\log_2\xi\ge D_{\min}^{\rm NL}$ and $\min_{U_A,U_B}R_{\rm g}\ge F_{\rm NL}^{-1}-1$.  The NMRE thus provides an explicit, spectral lower bound for the nonlocal versions of both, whose own optimizers remain unknown.

\section*{Discussion}

We have demonstrated that nonlocal magic is an exactly computable property of a many-body state, and that it measures the shape of the entanglement spectrum: the magic carried by quantum correlations is the distance of the spectrum from the nearest Bell-pair sector, and the same quantity sets a state's power to embezzle entanglement.  
The central lesson is that the entanglement spectrum, 
through the way its weight is spread across logarithmic rank scales,
contains information that neither entanglement entropy nor fixed-basis magic captures.
R\'enyi entropies of any index cannot replace the NMRE, since flat dyadic spectra have $D_{\min}^{\rm NL}=0$ at arbitrarily large entanglement. Nor can fixed-basis measures, which change under local rotations that leave the spectrum untouched.

Removing this basis dependence is what the local-unitary minimization was introduced for~\cite{Cao2025magical,Qian2025Nonlocal}, and earlier work identified its zero set and derived entropic bounds. Subsequent Schmidt-basis analyses obtained exact values at selected cuts and conjectured optimizing states~\cite{Huang2026Spectral, Torre2026Spectrum, Liu2026Schmidt, Franchini2026Schmidt}.  For the stabilizer fidelity, Theorem~\ref{thm:main} resolves the optimization at arbitrary bipartitions and proves that the sorted computational-basis representative is globally optimal, which certifies the nonlocal SRE within dimension-independent factor on every cut (Theorem~\ref{thm:SRE}).
Theorem~\ref{thm:octave} shows that the same dyadic organization of Schmidt weight governs universal LOCC embezzlement~\cite{ZanoniTheurerGour2024}: catalytic power is set not by how much entanglement a state carries but by how its Schmidt weight is spread across rank scales.
This is particularly consequential in many-body systems, where two families with identical logarithmic entanglement scaling behave in opposite ways: the nonlocal magic of critical chains grows without bound and they embezzle universally, whereas random-singlet states have no nonlocal magic at all.

Nonlocal magic, via NMRE, can now be computed exactly wherever the entanglement spectrum is available, from any matrix-product-state or density-matrix renormalization group simulation at no additional cost, and measured by entanglement spectroscopy with an explicit error  bound.  
This opens a new window on many-body physics: the magic carried by correlations can be tracked across phase diagrams and quantum phase transitions, followed in time after quenches, and contrasted between ergodic and nonergodic dynamics, in theory and experiment alike.
The window is not confined to condensed matter.  Theorem~\ref{thm:main}  applies to any bipartite pure state, and the NMRE quantifies nonlocal magic exactly wherever entanglement across a cut is meaningful, from gauge theories and particle and nuclear scattering to the electronic structure of molecules.

Several questions remain open.  Finding and proving the full local-unitary optimizers of the nonlocal SRE, stabilizer extent, generalized robustness and mana would determine how far the sorted computational basis geometry persists beyond the stabilizer fidelity.
For mixed states, no single spectrum determines a local-unitary orbit, and the right analogue of the spectral principle remains to be identified.  At conformal points, it remains to identify which CFT data fix the partial sums of the ordered Schmidt coefficients, to establish Eq.~\eqref{eq:CFT-main} beyond the smoothed Calabrese--Lefevre spectrum, and to test the predicted logarithmic scaling in finite critical chains.
Finally, sample-efficient estimators for partial sums of ordered Schmidt coefficients would connect these theoretical questions to scalable experiments.

\section*{Methods}

\subsection*{Schmidt-vector overlaps and flat stabilizer spectra}

We give a full matrix proof of Theorem~\ref{thm:main}. We choose product
bases and represent a bipartite vector by its coefficient matrix:
$
 \ket{\psi_X}=\sum_{a,b}(X)_{ab}\ket a_A\ket b_B.
 \label{eq:coefficient-matrix}
$
Its Schmidt amplitudes are the singular values of $X$.  If local unitaries act on the state, the coefficient matrix transforms as $X\mapsto U_AXU_B^{T}$.

\begin{lemma}[Schmidt-vector overlap]
\label{lem:Schmidt-overlap}
Let $\ket\psi$ and $\ket\phi$ be pure bipartite states with decreasing
Schmidt amplitudes $\alpha_0\ge\alpha_1\ge\cdots$ and
$\beta_0\ge\beta_1\ge\cdots$, respectively.  Then
\begin{equation}
 \max_{U_A,U_B}
 |\bra{\phi}(U_A\otimes U_B)\ket{\psi}|
 =\sum_x \alpha_x\beta_x.
 \label{eq:Schmidt-overlap}
\end{equation}
\end{lemma}

\begin{proof}
Writing $X,Y\in\mathbb C^{d_A\times d_B}$ for the two coefficient matrices, the modulus of the overlap is $|\Tr(Y^\dagger U_AXU_B^T)|$.  
Local unitaries preserve singular values, and the von Neumann trace inequality~\cite{vonNeumann1937, Mirsky1975, HornJohnson1991} bounds this modulus by the ordered singular-value scalar product $\sum_i s_i(X)s_i(Y)=\sum_x\alpha_x\beta_x$, where the sum runs over the smaller local dimension.
The bound is attained. Take full singular-value decompositions $X=W_1\Sigma_XW_2^\dagger$ and $Y=V_1\Sigma_YV_2^\dagger$, where $W_1,V_1\in U(d_A)$, $W_2,V_2\in U(d_B)$, and $\Sigma_X,\Sigma_Y\in\mathbb R^{d_A\times d_B}$ are rectangular diagonal matrices whose nonnegative diagonal entries are the singular values in nonincreasing order.
The choice $U_A=V_1W_1^\dagger$, $U_B=\bar V_2W_2^T$ gives $U_B^T=W_2V_2^\dagger$ and hence $Y^\dagger U_AXU_B^T
   =V_2\Sigma_Y^\dagger\Sigma_XV_2^\dagger$.
Its trace is $\Tr(\Sigma_Y^\dagger\Sigma_X)=\sum_x\alpha_x\beta_x$, proving attainability.
\end{proof}

We also use the standard bipartite normal form of pure stabilizer
states~\cite{Fattal2004}. Across the cut between subsystems $A$ and $B$, every such state can be brought by party-local Clifford circuits (which may contain entangling Clifford gates within $A$ and within $B$) to a tensor product of $k$ Bell pairs shared between the parties and $\ket0$ spectators on all remaining qubits, for some $0\le k\le\nu$. 
The Schmidt vector of such a stabilizer state is therefore $\bm u_{2^k}^{1/2}=2^{-k/2}(1,\ldots,1,0,\ldots,0)$ with $2^k$ nonzero entries. 
Conversely, every $k$ with $0\le k\le\nu$ occurs: the state $\ket{\Phi_k}=2^{-k/2}\sum_{x=0}^{2^k-1}\ket{x}_A\ket{x}_B$, which is $k$ Bell pairs on designated qubits and $\ket0$ elsewhere, is a stabilizer state of Schmidt rank $2^k$.

\subsection*{Proof of Theorem~\ref{thm:main}}

Lemma~\ref{lem:Schmidt-overlap} gives, for every stabilizer state
  $\ket{s}\in\cC_N\ket{0^N}$ of Schmidt rank $r=2^k$,
  \begin{equation}
   \max_{U_A,U_B}
   |\bra{s}(U_A\otimes U_B)\ket{\Psi}|^2
   =\frac1r\left(\sum_{x=0}^{r-1}\mu_x\right)^2.
   \label{eq:fixed-rank-overlap}
  \end{equation}
The maximum over local unitaries and the maximum over stabilizer states are maxima over a Cartesian product and may therefore be interchanged.
Consequently, using the bipartite stabilizer normal form and Eq.~\eqref{eq:fixed-rank-overlap} for each stabilizer state of Schmidt rank $2^k$, with $0\le k\le\nu$, we obtain
\begin{align}
   F_{\rm NL}(\Psi)
   &=
   \max_{\ket{s}\in\cC_N\ket{0^N}}
   \max_{U_A,U_B}
   |\bra{s}(U_A\otimes U_B)\ket{\Psi}|^2
   \notag\\
   &\le
   \max_{0\le k\le\nu}
   \frac{1}{2^k}
   \left(\sum_{x=0}^{2^k-1}\mu_x\right)^2.
   \label{eq:main-upper-bound}
\end{align}

For the converse, consider the ordered, zero-padded Schmidt decomposition
  \begin{equation}
  \ket{\Psi}
   =\sum_{x=0}^{2^\nu-1}
   \mu_x\ket{a_x}_A\ket{b_x}_B,
   \quad
   \mu_0\ge\mu_1\ge\cdots\ge\mu_{2^\nu-1}\ge0.
  \end{equation}
The zero-coefficient Schmidt vectors are chosen so that $\{\ket{a_x}\}_{x<2^\nu}$ and $\{\ket{b_x}\}_{x<2^\nu}$ are orthonormal families. We choose local unitaries $U_A^{\rm CB}$ and $U_B^{\rm CB}$ satisfying $U_A^{\rm CB}\ket{a_x}_A=\ket{x}_A$ and $U_B^{\rm CB}\ket{b_x}_B=\ket{x}_B$ for every $x<2^\nu$. 
On the smaller party this specifies the unitary on a full basis. If the parties have unequal dimensions, the map on the larger party is extended arbitrarily from the Schmidt support to a full orthonormal basis. With the convention $\mu_x\ge0$, these unitaries produce exactly the sorted computational-basis representative
  \begin{equation}
   \ket{\Psi_{\rm CB}}
   :=
   (U_A^{\rm CB}\otimes U_B^{\rm CB})\ket{\Psi}
   =
   \sum_{x=0}^{2^\nu-1}\mu_x\ket{x}_A\ket{x}_B,
   \label{eq:CB-representative-proof}
  \end{equation}
which therefore lies in the local-unitary orbit of $\ket{\Psi}$.

For each $k$, use the state $\ket{\Phi_k}$ defined above in the same  binary encoding. It consists of Bell pairs between the $k$ pairs of designated qubits carrying the least significant bits, with $\ket0$ on every other qubit, and is therefore a stabilizer state. Its squared overlap with $\ket{\Psi_{\rm CB}}$ is
\begin{equation}
  |\braket{\Phi_k|\Psi_{\rm CB}}|^2
  =\frac{1}{2^k}
   \left(\sum_{x=0}^{2^k-1}\mu_x\right)^2
  =F_k.
   \label{eq:CB-branch-overlap}
\end{equation}
Indeed, the support $\{0,\ldots,2^k-1\}$ of $\ket{\Phi_k}$ consists of the binary strings whose $\nu-k$ most significant bits vanish. These supports are nested initial segments of one common label ordering, so a single sorted representative serves every $k$.

Since every $\ket{\Phi_k}$ is a stabilizer state and $\ket{\Psi_{\rm CB}}$ lies in the local-unitary orbit of $\ket{\Psi}$, it follows that
\begin{align}
   F_{\rm NL}(\ket{\Psi})
   &\ge F_{\rm stab}(\ket{\Psi_{\rm CB}})
   \ge \max_{0\le k\le\nu} F_k
   \notag\\
   &=
   \max_{0\le k\le\nu}
   \frac{1}{2^k}
   \left(\sum_{x=0}^{2^k-1}\mu_x\right)^2.
   \label{eq:main-lower-bound}
\end{align}
Together with Eq.~\eqref{eq:main-upper-bound}, this proves Eq.~\eqref{eq:main-formula} and shows that $\ket{\Psi_{\rm CB}}$ maximizes the stabilizer fidelity over all choices of local unitaries $U_A,U_B$. Moreover, Eq.~\eqref{eq:CB-branch-overlap} attains the optimal rank-$2^k$ value of Eq.~\eqref{eq:fixed-rank-overlap} for every $k$. Thus the same computational-basis representative is simultaneously optimal in every allowed Schmidt-rank sector, with the corresponding stabilizer witness $\ket{\Phi_k}$ depending on $k$.

The only stabilizer-specific inputs are the flat dyadic Schmidt spectra and the existence of the witnesses $\ket{\Phi_k}$ for every $0\le k\le\nu$ across the selected cut, so the argument applies independently to every bipartition. The competition between neighboring branches $k$, which locates the sector walls, and the resulting cell decomposition of spectrum space, including the two-qubit family $\ket{\psi_p}$ of the main text, are described in Supplementary Section~\ref{sm:exact}.

\subsection*{Dyadic shells and upper bounds}

We group the sorted spectrum into the dyadic shells $D_0=\{0\}$ and $D_j=\{2^{j-1},\ldots,2^j-1\}$ for $1\le j\le\nu$, and define the shell weights $p_j:=\sum_{x\in D_j}\lambda_x$. 
Since $\lambda_x\ge\lambda_{2^{j-1}}$ for $x<2^{j-1}$, the branch of rank $2^{j-1}$ gives $F_{\rm NL}\ge2^{j-1}\lambda_{2^{j-1}}\ge p_j$ for $j\ge1$, while the rank-one branch gives $F_{\rm NL}\ge\lambda_0=p_0$. 
The $\nu+1$ shell weights sum to one, so $F_{\rm NL}\ge(\nu+1)^{-1}$, which is the universal bound given in the main text. For the upper bound, let $P_{\rm sh}:=\max_j p_j$. 
The Cauchy--Schwarz inequality gives $\sum_{x\in D_j}\mu_x\le2^{(j-1)/2}\sqrt{P_{\rm sh}}$ for $j\ge1$ and $\mu_0\le\sqrt{P_{\rm sh}}$. 
The support of the rank-$2^k$ branch is the union of the shells $D_0,\ldots,D_k$, so summing these bounds gives a geometric series, $\sum_{x<2^k}\mu_x\le\sqrt{P_{\rm sh}}\,\bigl(1+\sum_{j=1}^{k}2^{(j-1)/2}\bigr)$, and hence $F_k<(3+2\sqrt2)P_{\rm sh}$ for every $k$. Together, \begin{equation}
 P_{\rm sh}\le F_{\rm NL}
 \le(3+2\sqrt2)P_{\rm sh}.
 \label{eq:scale-min-entropy}
\end{equation}
Thus $D_{\min}^{\rm NL}$ differs by at most
$\log_2(3+2\sqrt2)$ from the min-entropy of Schmidt weight across fixed
dyadic scales.  Ordinary R\'enyi entropies cannot replace this
scale-resolved quantity, because flat dyadic spectra have
$D_{\min}^{\rm NL}=0$ at arbitrarily large entanglement.

\begin{proposition}[Entropy bound]
\label{prop:entropy-ceiling}
For every sorted Schmidt spectrum, with
$S_1=-\sum_x\lambda_x\log_2\lambda_x$, one has
$F_{\rm NL}\ge(2S_1+3)^{-1}$, or equivalently,
$D_{\min}^{\rm NL}\le\log_2(2S_1+3)$.
\end{proposition}

\begin{proof}
Because the spectrum is sorted, $(x+1)\lambda_x \le\sum_{y=0}^{x}\lambda_y \le1$, and hence $\lambda_x\le1/(x+1)$. Moreover, for $x\in D_k$, we have $\log_2(x+1)\ge(k-1)_+$, where $ (k-1)_+:=\max\{k-1,0\}$. Hence, 
\begin{equation}
   S_1=\sum_x\lambda_x\log_2\frac1{\lambda_x}
   \ge\sum_x\lambda_x\log_2(x+1)
   \ge\sum_{k=0}^{\nu}(k-1)_+\,p_k .
   \label{eq:entropy-shell-lower}
  \end{equation} 
  The shell bound proved above gives $p_k\le F_{\rm NL}$ for every $k$, while $\sum_kp_k=1$. Set $z:=F_{\rm NL}^{-1}$, $m:=\lfloor z\rfloor$ and $\delta:=z-m\in[0,1)$; the universal bound
  $F_{\rm NL}\ge(\nu+1)^{-1}$ ensures $z\le\nu+1$. Among all probability vectors with $p_k\le z^{-1}$, the linear cost $\sum_k(k-1)_+p_k$ is minimized by filling the lowest-cost shells
  first: weight $z^{-1}$ in each of the shells $k=0,\ldots,m-1$ and, when $\delta>0$, the remaining weight $\delta z^{-1}$ in shell $k=m$. This allocation costs $C_{\rm
  gr}=(m-1)(m+2\delta-2)/(2z)$, and a direct calculation gives $C_{\rm gr}-(z/2-3/2)=(2+\delta-\delta^2)/(2z)\ge0$. Consequently,
  \begin{equation}
   S_1\ge\sum_{k=0}^{\nu}(k-1)_+\,p_k\ge\frac{1}{2F_{\rm NL}}-\frac32 ,
   \label{eq:greedy-filling-bound}
  \end{equation}
  and rearranging gives $F_{\rm NL}\ge(2S_1+3)^{-1}$, which proves the proposition and Eq.~\eqref{eq:entropy-ceiling-main}.
  \end{proof}

\subsection*{Proof of Theorem~\ref{thm:octave}}
  
Theorem~\ref{thm:octave} is stated for arbitrary finite-dimensional states. A spectrum whose dimension is not a power of two is zero-padded to qubit registers, and we write $d=2^\nu$ for the dimension of the smaller party after padding. 
The value of $F_{\rm NL}$ is independent of the choice of embedding: any two embeddings into the same registers are related by party-local  unitaries, while enlarging the registers leaves the existing branches unchanged and adds only branches beyond the Schmidt rank. For those branches the numerator is constant, so their fidelities decrease inversely with rank.
For a dyadic rank $q$, write  $F^{[q]}:=q^{-1}\bigl(\sum_{x<q}\mu_x\bigr)^2$, and, for an integer $\ell\ge0$, let $\eta_\ell:=\sum_{x=\ell}^{2\ell}\lambda_x$ denote the weight of the octave $[\ell,2\ell]$. Thus $\eta_{\rm emb}=\max_{\ell\in\mathbb Z_{\ge0}}\eta_\ell$. Assume $d\ge2$, the case $d=1$ being immediate. For $d/2\le\ell\le d-1$, zero padding gives $\eta_\ell=\sum_{x=\ell}^{d-1}\lambda_x$, and
  $
   \eta_\ell
   \le\sum_{x=\ell-1}^{d-2}\lambda_x
   \le\eta_{\ell-1}.
  $  
The first inequality follows from sortedness, while the second follows from $\{\ell-1,\ldots,d-2\}\subseteq \{\ell-1,\ldots,2\ell-2\}$. Since $\eta_\ell=0$ for $\ell\ge d$, the maximum is attained at some $0\le\ell\le d/2-1$, for which $2\ell+1<d$.

  For $\ell=0$ the rank-one branch equals $\lambda_0=\eta_0$. For $1\le\ell\le d/2-1$, let $q$ be the smallest power of two with $q\ge2\ell+1$, so that $q\le\min\{4\ell,d\}$. If
  $\mu_\ell>0$, sortedness gives $\sum_{x<\ell}\mu_x\ge\ell\mu_\ell$ and $\eta_\ell\le\mu_\ell\sum_{x=\ell}^{2\ell}\mu_x$; since the first $q$ coefficients include both sums,
  \begin{equation}
   \sum_{x=0}^{q-1}\mu_x\ge\ell\,\mu_\ell+\frac{\eta_\ell}{\mu_\ell}\ge2\sqrt{\ell\,\eta_\ell}
   \label{eq:octave-AMGM}
  \end{equation}
  by the AM--GM inequality, and hence $F^{[q]}\ge4\ell\,\eta_\ell/q\ge\eta_\ell$. If $\mu_\ell=0$ the octave has zero mass. Maximizing over $\ell$ proves $\eta_{\rm emb}\le F_{\rm NL}$.

Conversely, each dyadic shell $D_j=\{a,\ldots,2a-1\}$ with $a=2^{j-1}$ lies inside the octave $\{a,\ldots,2a\}$, and $D_0$ is the octave $\ell=0$, so $P_{\rm sh}\le\eta_{\rm emb}$. Substituting this into Eq.~\eqref{eq:scale-min-entropy} proves the left-hand bound of Theorem~\ref{thm:octave}. The constant is asymptotically sharp: a family of spectra with $F_{\rm NL}/\eta_{\rm emb}\to3+2\sqrt2$ is given in Supplementary Section~\ref{sm:properties}.

\subsection*{Haar-random spectra}
For a Haar-random state on $\mathbb C^d\otimes\mathbb C^D$, with $d=2^\nu\le D$ and $d/D\to c\in(0,1]$, the rescaled eigenvalues $\xi_x=d\lambda_x$ are distributed according to the  Marchenko--Pastur law, and for a typical state, i.e. with probability tending to one as $d\to\infty$, the largest one approaches the upper edge, $\xi_0\to(1+\sqrt  c)^2$~\cite{marvcenko1967distribution,Nadal2011,BaiYin1993}. For each fixed dyadic relative rank $\alpha=2^{-j}$, the corresponding branch of Eq.~\eqref{eq:main-formula} approaches  $\Phi_c(\alpha):=I_c(\alpha)^2/\alpha$ for a typical state, where $I_c(\alpha)$ is the contribution of the largest fraction $\alpha$ of the eigenvalues to the mean of $\sqrt x$ under  the Marchenko--Pastur law, so that $I_c(1)$ is the full mean. The ordering of the eigenvalues and the convergence of the upper edge control the small-rank branches uniformly, so that for a typical state $F_{\rm NL}\to\max_{j\ge0}\Phi_c(2^{-j})$.

The Cauchy--Schwarz inequality gives $\Phi_c(\alpha)\le\alpha+\sqrt{c\alpha(1-\alpha)} \le(1+\sqrt3)/4$ for $\alpha\le\tfrac14$, whereas $\Phi_c(1)\ge\Phi_1(1)=64/(9\pi^2)$. Thus only the half-rank and full-rank branches can maximize the envelope. Their direct comparison shows that the half-rank branch wins at $c=1$, while the full-rank branch wins for $c\le\tfrac12$. Consequently, at the fixed-imbalance ratios $c=2^{-\Delta}$,
  \begin{equation}
   F_{\rm NL}
   \xrightarrow{\Pr}
   \begin{cases}
    \dfrac{128}{9\pi^{2}}\sin^{6}\!\varphi_\star,
    & \Delta=0,\\[7pt]
    \left[
     {}_2F_1\!\left(-\dfrac12,\dfrac12;2;2^{-\Delta}\right)
    \right]^2,
    & \Delta\ge1,
   \end{cases}
   \label{eq:Haar-limit-main}
  \end{equation}
where $\varphi_\star\in(0,\pi/2)$ is the unique solution of $2\varphi_\star-\sin2\varphi_\star=\pi/2$. Taking $-\log_2$ gives the limit $D_\Delta$ of Proposition~\ref{prop:Haar}, and $-\ln\Phi_c(1)=c/4+O(c^2)$ gives its large-imbalance expansion. Details are given in Supplementary Section~\ref{sm:Haar}.

\subsection*{Critical spectra}
The smoothed Calabrese--Lefevre ansatz~\cite{Calabrese2008spectrum} assumes $R_\alpha=\Tr\rho_A^\alpha=e^{-b(\alpha-1/\alpha)}$ with $b=-\ln\lambda_{\max}$, so that $S_1=2b/\ln2$. In terms of the depth $u=\ln(\lambda_{\max}/\lambda)$, the number of eigenvalues with depth at most $u$ is $N_b(u)=I_0(2\sqrt{bu})$, and the sum of their Schmidt amplitudes $\sqrt{\lambda}$ is $G_b(u)=e^{-b/2}\bigl[1+\int_0^u e^{-v/2}\sqrt{b/v}\, I_1(2\sqrt{bv})\,dv\bigr]$, where $I_0$ and $I_1$ are modified Bessel functions. 
Treating the rank continuously, set $r=N_b(u)$; the resulting fidelity envelope is $G_b(u)^2/r$, sampled at $r=2^k$, $0\le k\le\nu$. Writing $u=bw^2$, Bessel asymptotics and an endpoint Laplace approximation give, uniformly on compact subsets of $0<w<2$,
  \begin{equation}
   \frac{G_b(bw^2)^2}{N_b(bw^2)}
   =\frac{2}{\sqrt{\pi b}}\,\frac{w^{-1/2}}{(2-w)^2}\,
   e^{-b(w-1)^2}\bigl[1+O(b^{-1})\bigr].
   \label{eq:CL-envelope-methods}
  \end{equation}
Explicit bounds show that the envelope is exponentially smaller for $w\le\tfrac14$ and $w\ge\tfrac74$. Together with these bounds, a differentiable uniform expansion near $w=1$ places the maximum at $w_\star=1+\tfrac{3}{4b}+O(b^{-2})$, with value $2/\sqrt{\pi b}\,[1+O(b^{-1})]$. Rounding to the nearest dyadic rank shifts $\ln r$ by at most $\tfrac12\ln2$, which changes the value only by a relative $O(b^{-1})$. Hence $F_{\rm NL}=2/\sqrt{\pi b}\,[1+O(b^{-1})]$ at $\ln r_\star=2b-\tfrac12\ln b+O(1)$. For a bulk interval, $b=(c_{\rm CFT}/6)\ln L_{\rm eff}+O(1)$~\cite{CalabreseCardy2004}, and therefore $r_\star=e^{O(1)}L_{\rm eff}^{c_{\rm CFT}/3}/ \sqrt{\ln L_{\rm eff}}$. This polynomial rank lies below the exponentially growing local Hilbert-space cutoff for large spin blocks, and the preceding asymptotics yield Eq.~\eqref{eq:CFT-main}. The result is conditional on the smoothed ansatz, fixed-index CFT scaling alone is insufficient (Supplementary Section~\ref{sm:CL}).

\subsection*{Numerical evaluation}

Results in Fig.~\ref{fig:classes} follow from Eq.~\eqref{eq:main-formula} applied to explicitly known spectra.  For the Haar points we sampled normalized complex Gaussian vectors at balanced cuts of
$N=8,10,\ldots,16$ qubits,
and averaged $S_1$ and $D_{\min}^{\rm NL}$ over $150$ to $4000$ samples.
The van Dam--Hayden values use the harmonic spectrum of $\ket{\Gamma_d}$ for $d=2^\nu$ up to $\nu=2900$, with the partial sums $\sum_{i\le2^k}i^{-1/2}$ and $H_d$ evaluated from their Euler--Maclaurin expansions.
The tensor-power values use the product spectrum of $\bm p=(\tfrac12,\tfrac3{10},\tfrac15)$ for $m\le10^3$, with exact integer bookkeeping of the degeneracies in logarithmic space.

\section*{Acknowledgements}
P.S. thanks Xhek Turkeshi for collaboration on related topics, for a careful reading of the manuscript and for valuable feedback.
P.S. thanks Poetri Sonya Tarabunga bringing Ref.~\cite{Tarabunga2025BMSA} to our attention.
P.S. acknowledges fellowship within the Generaci\'on D initiative,
Red.es, Ministerio para la Transformaci\'on Digital y de la Funci\'on
P\'ublica, for talent attraction (C005/24-ED CV1), funded by the European
Union NextGenerationEU funds through PRTR.

\emph{Note added.---}After the completion of this manuscript, we became aware of a closely related independent work by L. Leone and A. Hamma~\cite{LeoneHammaToAppear}. Where the results overlap, they are consistent.

%

\pagebreak 

\onecolumngrid
\setcounter{secnumdepth}{3}
\begin{center}
{\large\bfseries Supplementary Information for\\[4pt] \papertitle}\\[10pt]
Piotr Sierant
\end{center}
\vspace{6pt}
\twocolumngrid 

\setcounter{equation}{0}
\setcounter{section}{0}
\setcounter{theorem}{0}
\setcounter{proposition}{0}
\setcounter{corollary}{0}
\setcounter{lemma}{0}
\renewcommand{\theequation}{S\arabic{equation}}
\renewcommand{\thesection}{S\arabic{section}}
\renewcommand{\thetheorem}{S\arabic{theorem}}
\renewcommand{\theproposition}{S\arabic{proposition}}
\renewcommand{\thecorollary}{S\arabic{corollary}}
\renewcommand{\thelemma}{S\arabic{lemma}}
\setcounter{figure}{0}
\renewcommand{\thefigure}{S\arabic{figure}}
\renewcommand{\theHfigure}{supp.fig.\arabic{figure}}
\renewcommand{\theHsection}{supp.\arabic{section}}
\renewcommand{\theHtheorem}{supp.\arabic{theorem}}
\renewcommand{\theHproposition}{supp.\arabic{proposition}}
\renewcommand{\theHcorollary}{supp.\arabic{corollary}}
\renewcommand{\theHlemma}{supp.\arabic{lemma}}

\noindent
This Supplementary Information supplies additional proofs and extensions of the results stated in the main text and in Methods.
Section~\ref{sm:exact} describes the Bell-sector cell decomposition of spectrum space behind Theorem~\ref{thm:main}.
Section~\ref{sm:properties} gives the stability identities and the tensor-product properties of the NMRE, including the proof of the tensor-power dichotomy.
Section~\ref{sm:sre} proves the arbitrary-cut comparison with the nonlocal SRE, derives the linear equivalence on arbitrary cuts, and the sharper balanced-cut constants.
Section~\ref{sm:Haar} proves the Haar asymptotics of Proposition~\ref{prop:Haar}; Sec.~\ref{sm:CL} treats critical spectra. Section~\ref{sm:embezzlement} derives finite-error embezzlement bounds, evaluates the harmonic catalysts, and separates LOCC from communication-free embezzlement. Section~\ref{sm:qudits} extends the exact solution to qudits.

\section{Bell-pair sectors}
\label{sm:exact}

\begin{figure}[t]
 \centering
 \includegraphics[width=0.46\textwidth]{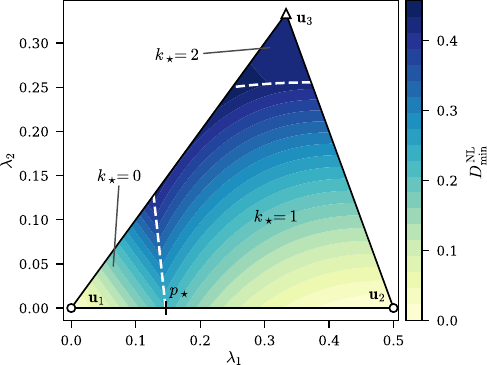}
 \caption{Exact Bell-sector diagram for rank-3 Schmidt spectra
 $\lambda_0\ge\lambda_1\ge\lambda_2$, evaluated from
 Eq.~\eqref{eq:main-formula}: heat map of $D_{\min}^{\rm NL}$,
 with the optimal-sector cells $k_\star=0,1,2$ of
 Eq.~\eqref{eq:main-formula} separated by the white walls, across
 which the measure is continuous but not analytic.  The zero set
 consists of $\bm u_1$ and $\bm u_2$ (circles); the flat rank-3
 spectrum $\bm u_3$ (triangle) has $D_{\min}^{\rm NL}=\log_2(4/3)$.
 The bottom edge is the two-qubit family $\ket{\psi_p}$ of the main text
 with the sector wall at $p_\star=\tfrac12+\tfrac{\sqrt2}{4}$ (tick).}
 \label{fig:sector-simplex}
\end{figure}

  Theorem~\ref{thm:main} is proved in Methods. Here we record how the optimal Bell-pair sector changes across spectrum space. Neighboring branches $F_k$ of Eq.~\eqref{eq:main-formula} obey
  \begin{equation}
   F_{k+1}\ge F_k
   \quad\Longleftrightarrow\quad
   \sum_{x=2^k}^{2^{k+1}-1}\mu_x\ge(\sqrt2-1)\sum_{x=0}^{2^k-1}\mu_x ,
   \label{sm:eq:branch-wall}
  \end{equation}
  i.e. the next dyadic shell must carry at least a fraction $\sqrt2-1$ of the amplitude already accumulated. For the two-qubit family $\ket{\psi_p}$ of the main text this reproduces the
  wall at $p_\star=\tfrac12+\tfrac{\sqrt2}{4}$. The branches need not increase up to $k_\star$ and decrease afterwards. For instance, for $\bm\lambda\propto(1,0.16,0.16,0.16)$ one finds $F_0>F_1<F_2$, so Eq.~\eqref{sm:eq:branch-wall} separates two sectors only where no third branch dominates. Figure~\ref{fig:sector-simplex} shows the resulting cell decomposition and
  the landscape of $D_{\min}^{\rm NL}$ for all rank-3 spectra.

\section{Stability identities and tensor products}
\label{sm:properties}

\subsection{Distance and stability identities}

For an allowed rank $r=2^k$, let $\bar\mu_r:=r^{-1}\sum_{x<r}\mu_x$ be
the mean of the leading amplitudes, so that the branch is
$F^{[r]}=r\bar\mu_r^2$. Orthogonal projection of $\bm\mu$ onto the flat unit vector $\bm u_{2^k}^{1/2}$ gives
\begin{equation}
 1-F^{[r]}
 =\sum_{x\ge r}\lambda_x+
 \sum_{x=0}^{r-1}(\mu_x-\bar\mu_r)^2,
 \label{sm:eq:tail-flatness}
\end{equation}
so that $1-F_{\rm NL}$ is the minimum over dyadic ranks of the Schmidt weight outside the support plus the nonflatness inside it. The two terms are the tail-dominated and flatness-dominated regimes discussed for the two-qubit family in the main text.
If $\cA_{\rm NL}$ denotes the union of the local-unitary orbits of all pure stabilizer states, then the pure-state trace-distance identity $\tfrac12\bigl\|\ket\psi\!\bra\psi-\ket\phi\!\bra\phi\bigr\|_1=\sqrt{1-|\braket{\psi|\phi}|^2}$, the nearest state in $\cA_{\rm NL}$ lies at trace distance 
\begin{equation}
 \min_{\phi\in\cA_{\rm NL}}\frac12
 \left\|
 \ket{\Psi}\!\bra{\Psi}-\ket{\phi}\!\bra{\phi}
 \right\|_1
 =\sqrt{1-F_{\rm NL}}.
 \label{sm:eq:orbit-rigidity}
\end{equation}

\paragraph{Spectral estimation error.}
The value of $F_{\rm NL}$ given by Eq.~\eqref{eq:main-formula} is robust to errors in the entanglement spectrum, such as those introduced by entanglement spectroscopy or by truncation in tensor-network simulations. 
Let $\hat{\bm\lambda}$ be a sorted estimate with $\|\hat{\bm\lambda}-\bm\lambda\|_1\le\varepsilon$, and let $\hat F_{\rm NL}$ be the corresponding value of the maximum in Eq.~\eqref{eq:main-formula}.  
Since $(\sqrt a-\sqrt b)^2\le|a-b|$ for $a,b\ge0$, the amplitude vectors obey $\|\hat{\bm\mu}-\bm\mu\|_2\le\sqrt{\varepsilon}$; every branch is the squared overlap of $\bm\mu$ with a fixed unit vector, so $|\sqrt{\hat F_{\rm NL}}-\sqrt{F_{\rm NL}}|\le\sqrt{\varepsilon}$ and
\begin{equation}
 |\hat F_{\rm NL}-F_{\rm NL}|\le2\sqrt{\varepsilon}.
 \label{sm:eq:spectral-robustness}
\end{equation}
Moreover, by Eq.~\eqref{eq:scale-min-entropy}, the dyadically binned spectrum, i.e., the shell weights $p_k$ alone, determines $D_{\min}^{\rm NL}$ within the additive constant $\log_2(3+2\sqrt2)$, which is asymptotically sharp (see below). Once a maximizing sector $k_\star$ is known, its exact value depends only on the leading $2^{k_\star}$ Schmidt probabilities.

\paragraph{Sharpness of the constant $3+2\sqrt2$.}
The constant in Eq.~\eqref{eq:scale-min-entropy}, and hence in Theorem~\ref{thm:octave}(i), cannot be improved. With $d=2^\nu$, assign the unnormalized weights
  \begin{equation}
   a_x=2^{-m}\quad\text{for}\quad 2^m-1\le x\le2^{m+1}-2,
   \label{sm:eq:octave-sharp-family}
  \end{equation}
with $m=0,\ldots,\nu-1$,  set $a_{d-1}=0$, and normalize. Each block $m$ consists of $2^m$ entries of weight $2^{-m}$, so every dyadic shell and every octave $[\ell,2\ell]$ carries unnormalized mass at most one, and $P_{\rm sh}=\eta_{\rm emb}=1/\nu$. The rank-$d$ branch gives
\begin{equation}
   \frac{F^{[d]}}{\eta_{\rm emb}}=\frac1d\Bigl(\sum_{m=0}^{\nu-1}2^{m/2}\Bigr)^2
   =\frac{(2^{\nu/2}-1)^2}{d\,(\sqrt2-1)^2}\longrightarrow3+2\sqrt2 .
   \label{sm:eq:octave-sharp-limit}
\end{equation}
Since $F^{[d]}\le F_{\rm NL}\le(3+2\sqrt2)P_{\rm sh}=(3+2\sqrt2)\eta_{\rm emb}$, both ratios $F_{\rm NL}/P_{\rm sh}$ and $F_{\rm NL}/\eta_{\rm emb}$ tend to $3+2\sqrt2$ as $\nu\to\infty$.

\subsection{Tensor-power flatness dichotomy}
\label{sm:ssec:tensor-dichotomy}
Let $\ket\psi=\ket{\psi_{\bm p}}:=\sum_{i=1}^r\sqrt{p_i}\ket{i}_A\ket{i}_B$ have positive Schmidt probabilities $\bm p=(p_1,\ldots,p_r)$, entanglement entropy $H:=-\sum_ip_i\ln p_i$ and capacity of entanglement $V:=\sum_ip_i(\ln p_i)^2-H^2$, the variance of $-\ln p_i$ when $i$ is drawn with probability $p_i$; $V=0$ exactly when all $p_i$ are equal. All copies on side $A$ are regrouped against all copies on side $B$, so the local-unitary optimization may act collectively on each party.

\begin{theorem}[Tensor-power flatness dichotomy]
\label{sm:thm:tensor-dichotomy}
If $V>0$, then
\begin{equation}
 \begin{aligned}
 F_{\rm NL}(\psi^{\otimes m})&=\Theta_{\psi}(m^{-1/2}),\\
 D_{\min}^{\rm NL}(\psi^{\otimes m})
 &=\frac12\log_2m+O_{\psi}(1).
 \end{aligned}
 \label{sm:eq:tensor-nonflat-law}
\end{equation}
If $V=0$, the positive spectrum is flat at some rank $r$, and exactly
\begin{equation}
 \begin{aligned}
 F_{\rm NL}(\psi^{\otimes m})
 &=2^{-\operatorname{dist}(m\log_2r,\mathbb Z)},\\
 D_{\min}^{\rm NL}(\psi^{\otimes m})
 &=\operatorname{dist}(m\log_2r,\mathbb Z)\le\frac12.
 \end{aligned}
 \label{sm:eq:tensor-flat-law}
\end{equation}
Here $\operatorname{dist}(x,\mathbb Z)$ is the distance to the nearest
integer.  Dyadic $r$ gives $D_{\min}^{\rm NL}=0$ for every $m$; otherwise the bounded
value generally oscillates with $m$.
Consequently, the tensor powers of $\psi$ form a universal LOCC-embezzling
family if and only if its positive Schmidt spectrum is nonflat.
\end{theorem}

\begin{proof}
Let $I_1,\ldots,I_m$ be independent indices drawn from $\bm p$ and set $S_m:=-\sum_{j=1}^m\ln p_{I_j}$, so that $e^{-S_m}$ is a product Schmidt probability drawn with its own weight. $S_m$ has mean $mH$ and variance $mV$. Suppose first that $V>0$. Choose $h>0$ smaller than the minimum nonzero spacing between distinct values of $-\ln p_i$, and define the concentration  function
  \begin{equation}
   Q_m(h):=\sup_{\tau}\Pr\{\tau-h<S_m\le \tau\}.
   \label{sm:eq:tensor-concentration-function}
  \end{equation}
  The Kolmogorov--Rogozin inequality~\cite{Kesten1969} gives $Q_m(h)=O_\psi(m^{-1/2})$, because each summand $-\ln p_{I_j}$ has a nondegenerate distribution when $V>0$.

  Let $\lambda_0^{(m)}\ge\lambda_1^{(m)}\ge\cdots$ be the ordered product spectrum and put
  \begin{equation}
   \mathcal N_m(\tau):=\#\{x:-\ln\lambda_x^{(m)}\le \tau\}.
   \label{sm:eq:tensor-counting-function}
  \end{equation}
  The exact change of measure
  \begin{equation}
   e^{-\tau}\mathcal N_m(\tau)
   =\mathbb E\!\left[e^{S_m-\tau}\bm 1_{\{S_m\le \tau\}}\right]
   \le\frac{Q_m(h)}{1-e^{-h}}
   \label{sm:eq:tensor-counting-bound}
  \end{equation}
  follows by partitioning $(-\infty,\tau]$ into intervals of width $h$. For a nonzero dyadic shell with $a$ entries beginning at label $a$, take $\tau=-\ln\lambda_{a}^{(m)}$. Sortedness
  and $a+1\le\mathcal N_m(\tau)$ then give
  \begin{equation}
   \sum_{x=a}^{2a-1}\lambda_x^{(m)}
   \le a\lambda_{a}^{(m)}
   =O_\psi(m^{-1/2}).
   \label{sm:eq:tensor-shell-upper}
  \end{equation}
  The rank-one shell is exponentially smaller, so the largest dyadic-shell weight $P_{\rm sh}^{(m)}$ is $O_\psi(m^{-1/2})$.

For the matching lower bound, let $R_m$ be the random rank obtained by sampling the ordered spectrum, so $\Pr(R_m=x)=\lambda_x^{(m)}$ and $S_m=-\ln\lambda_{R_m}^{(m)}$. Sortedness implies $(R_m+1)\lambda_{R_m}^{(m)}\le1$, hence $\ln(R_m+1)\le S_m$. With $\mu_m=mH$ and $t=A\sqrt m$, Chebyshev's inequality gives
   \begin{align}
   \Pr[\ln(R_m+1)\ge\mu_m+t]&\le\frac{V}{A^2},
   \label{sm:eq:tensor-rank-upper-tail}\\
   \Pr[\ln(R_m+1)\le\mu_m-t]&\le\frac{4V}{A^2}+e^{-t/2}.
   \label{sm:eq:tensor-rank-lower-tail}
  \end{align}
The second bound follows by splitting the event according to whether $S_m\le\mu_m-t/2$, which Chebyshev's inequality bounds by $4V/A^2$.
On the complementary event, the rank lies among at most  $e^{\mu_m-t}$ labels whose probabilities are each smaller than $e^{-\mu_m+t/2}$, contributing at most $e^{-t/2}$. Choosing a fixed sufficiently large $A$ therefore places, for all sufficiently large  $m$, a positive, $m$-independent probability in $|\ln(R_m+1)-mH|\le A\sqrt m$.
This interval meets only $O_\psi(\sqrt m)$
  dyadic shells, so one of them carries $\Omega_\psi(m^{-1/2})$ probability. Therefore $P_{\rm sh}^{(m)}=\Theta_\psi(m^{-1/2})$, and Eq.~\eqref{eq:scale-min-entropy} proves
  Eq.~\eqref{sm:eq:tensor-nonflat-law}.

Finally, $V=0$ holds exactly when all positive $p_i$ are equal. The $m$-copy spectrum is then flat at rank $R=r^m$. A dyadic branch of rank $q$ equals $q/R$ for $q\le R$ and $R/q$ for $q\ge R$. Maximizing over the powers of two immediately gives Eq.~\eqref{sm:eq:tensor-flat-law}. The statement follows from Eq.~\eqref{eq:embezzlement-iff-main}.
  \end{proof}

\subsection{Tensor identities and the LOCC caveat}

If $\ket{\phi_{\rm free}}$ has a flat Schmidt spectrum of rank $s=2^\ell$, its product with $\ket\Psi$ repeats each Schmidt amplitude $s$ times with a factor $s^{-1/2}$: for $q<s$ the branch of rank $q$ is no larger than the rank-one branch of $\ket\Psi$, for $s\le q\le s2^\nu$ writing $q=sr$ reproduces the branch of rank $r$ of $\ket\Psi$, and above the support  increasing $q$ only lowers the saturated branch. Hence
  \begin{equation}
   F_{\rm NL}(\ket\Psi\otimes\ket{\phi_{\rm free}})=F_{\rm NL}(\ket\Psi).
   \label{sm:eq:free-factor}
  \end{equation}
If $\ket{s}$ and $\ket{s'}$ are the stabilizer states closest to $\ket\Psi$ and $\ket\Phi$ in their optimal local frames, then $\ket{s}\otimes\ket{s'}$ is a stabilizer state in a local frame of $\ket\Psi\otimes\ket\Phi$, so
  \begin{equation}
   F_{\rm NL}(\ket\Psi\otimes\ket\Phi)\ge F_{\rm NL}(\ket\Psi)\,F_{\rm NL}(\ket\Phi),
   \label{sm:eq:supermultiplicativity}
  \end{equation}
i.e. the NMRE is subadditive, $D_{\min}^{\rm NL}(\ket\Psi\otimes\ket\Phi)\le D_{\min}^{\rm NL}(\ket\Psi)+D_{\min}^{\rm NL}(\ket\Phi)$, and the inequality can be strict: for the flat  rank-three spectrum $\bm u_3$ one has $D_{\min}^{\rm NL}(\bm u_3)=\log_2\tfrac43$, whereas two copies have the flat rank-nine spectrum $\bm u_9$ and $D_{\min}^{\rm NL}(\bm  u_9)=\log_2\tfrac98<2\log_2\tfrac43$.
The bound $D_{\min}^{\rm NL}\le\log_2(\nu+1)$, applied to $\ket\Psi^{\otimes t}$ with $t\nu$ qubits on each side, gives
  \begin{equation}
   0\le\frac1tD_{\min}^{\rm NL}(\ket\Psi^{\otimes t})
   \le\frac{\log_2(t\nu+1)}{t}\longrightarrow0.
   \label{sm:eq:zero-density}
\end{equation}
Finally, the deterministic LOCC conversion of a Bell pair into $\ket{\psi_p}=\sqrt p\,\ket{00}+\sqrt{1-p}\,\ket{11}$ invoked in the main text is realized by a local measurement with Kraus operators $\operatorname{diag}(\sqrt p,\sqrt{1-p})$ and $\operatorname{diag}(\sqrt{1-p},\sqrt p)$, followed by $X_A\otimes X_B$ in the second outcome; since $F_{\rm NL}(\ket{\psi_p})=\max\{p,\tfrac12+\sqrt{p(1-p)}\}<1$ for $\tfrac12<p<1$, the NMRE is not monotone under LOCC.

\section{Comparison with nonlocal SRE}
\label{sm:sre}

\subsection{Arbitrary cuts}

We denote the nonlocal stabilizer purity as $P_{2,{\rm NL}}(\ket\Psi):=\max_{U_A,U_B}Q_N\bigl((U_A\otimes U_B)\ket\Psi\bigr)=2^{-\cM_{\rm NL}^{\rm SRE}(\ket\Psi)}$. Let $\ket{s}\in\cC_N\ket{0^N}$ be a stabilizer state with signed stabilizer group $\mathsf S_s$, whose $2^N$ elements have $\ket{s}$ as a common $+1$ eigenvector. The projector identity and convexity give
  \begin{align}
   f_s:=|\braket{s|\psi}|^2
   &=2^{-N}\sum_{Q\in\mathsf S_s}\langle Q\rangle_\psi,
   \notag\\
   f_s^4
   &\le2^{-N}\sum_{Q\in\mathsf S_s}\langle Q\rangle_\psi^4
   \le Q_N(\ket\psi).
   \label{sm:eq:fid-fourth}
\end{align}
Maximizing over $\ket s$ proves $F_{\rm stab}(\ket\psi)^4\le Q_N(\ket\psi)$~\cite{Haug2023Monotones}, and the testing inequality of Ref.~\cite{Haug2024Algorithms} is $F_{\rm stab}(\ket\psi)\ge2Q_N(\ket\psi)-1$. Evaluating the first inequality in a fidelity-optimal local frame and maximizing the second over local frames gives
  \begin{equation}
   F_{\rm NL}^4\le P_{2,{\rm NL}}\le\frac{1+F_{\rm NL}}2,
   \label{sm:eq:purity-bracket}
  \end{equation}
and taking negative logarithms proves the upper bound and the saturating lower bound in Eq.~\eqref{eq:SRE-fidelity-main}.
The coefficient four is sharp: for the two-qubit family $\ket{\psi_p}$, local unitaries cannot increase the Pauli fourth moment beyond its value in the Schmidt frame, $P_{2,{\rm NL}}(\ket{\psi_p})=1-t^2+t^4$ with $t=2\sqrt{p(1-p)}$, so that $\cM_{\rm NL}^{\rm SRE}/D_{\min}^{\rm NL}\to4$ both as $p\to1$ (where $\cM\simeq4\epsilon/\ln2$, $D\simeq\epsilon/\ln2$ with $\epsilon=1-p$) and as $p\to\tfrac12$ (where $\cM\simeq4\delta^2/\ln2$, $D\simeq\delta^2/\ln2$ with $\delta=p-\tfrac12$).

\paragraph{Linear lower bound on every cut.}
The inverse theorem for the Gowers-3 norm of quantum states~\cite{ArunachalamBravyiDutt2024} states that $F_{\rm stab}(\ket\psi)\ge Q_N(\ket\psi)^{C_\star}$ for a universal constant $C_\star>1$, since ${\rm Gowers}(\ket\psi,3)^8=Q_N(\ket\psi)$ (Lemma 3.3 of Ref.~\cite{ArunachalamDutt2024}); equivalently, $D_{\min}\le C_\star\bigl(-\log_2Q_N\bigr)$ in every frame~\cite{Tarabunga2025BMSA}. A closely related inverse bound, for the Bell-difference moment $\gamma\in[Q_N^2,Q_N]$ and with an explicit exponent, is proved independently in Ref.~\cite{BaoDordrechtHelsen2025}: $F_{\rm stab}\ge c\,\gamma^{K}$ with a universal prefactor $c$ and $K=112$ in the arXiv version, i.e. $D_{\min}\le2K\bigl(-\log_2Q_N\bigr)+O(1)$. Evaluating the multiplicative form in the frame $U_A^\star\otimes U_B^\star$ that maximizes $Q_N$ gives
\begin{equation}
 D_{\min}^{\rm NL}\le D_{\min}\bigl((U_A^\star\otimes U_B^\star)\ket\Psi\bigr)\le C_\star\,\cM_{\rm NL}^{\rm SRE},
 \label{sm:eq:gowers-linear}
\end{equation}
which proves the linear lower bound in Eq.~\eqref{eq:SRE-fidelity-main}.

\subsection{Balanced cuts}

For a balanced cut $n_A=n_B=n$, we define the shell-max weights $\widehat w_0:=\lambda_0$ and $\widehat w_k:=2^{k-1}\lambda_{2^{k-1}}$ for $1\le k\le n$, and $Q_{\rm sh}:=\sum_{k=0}^n\widehat w_k^4$. 
On balanced cuts one obtains the explicit bound $\cM_{\rm NL}^{\rm SRE}\ge\tfrac15D_{\min}^{\rm NL}$, which sharpens the linear bound in Eq.~\eqref{eq:SRE-fidelity-main} since $C_\star\ge5$. The key input is the following result of the companion work~\cite{Sierant2026CBOptimality}, which we restate in the present notation.

\begin{theorem}
\label{sm:thm:shell-sandwich}
For a balanced cut $n_A=n_B$, with $P_{2,{\rm CB}}$ the stabilizer purity of the computational basis representative $\ket{\Psi_{\mathrm{CB}}}$,
\begin{equation}
 \begin{aligned}
 \frac1{16}Q_{\rm sh}
 &\le P_{2,{\rm CB}}
 \le P_{2,{\rm NL}}
 \le C_{\rm sh}Q_{\rm sh},\\
 C_{\rm sh}&=198+576(1+\sqrt2),
 \end{aligned}
 \label{sm:eq:shell-purity-lemma}
\end{equation}
uniformly over the local-unitary orbit.
\end{theorem}

Two consequences of the ordering $\lambda_0\ge\lambda_1\ge\cdots$ connect the shell-max weights to the nonlocal stabilizer fidelity.  First, $F_{\rm NL}\ge\widehat w_k$ for every $k$: the $k=0$ case is the rank-one branch, and for $k\ge1$ the rank-$2^{k-1}$ branch obeys
\begin{equation}
 F_{\rm NL}\ge
 2^{-(k-1)}\Bigl(\sum_{x=0}^{2^{k-1}-1}\mu_x\Bigr)^2
 \ge2^{k-1}\lambda_{2^{k-1}}=\widehat w_k .
 \label{sm:eq:fid-dominates-shell}
\end{equation}
Second, $\sum_k\widehat w_k\le2$, because $\widehat w_1\le\lambda_0$ and,
for $k\ge2$, the shell $\{2^{k-2},\ldots,2^{k-1}-1\}$ carries
probability at least $\widehat w_k/2$.  Hence
$P_{2,{\rm NL}}\le C_{\rm sh}(\max_k\widehat w_k)^3\sum_k\widehat w_k
\le2C_{\rm sh}F_{\rm NL}^3$, that is, with
$x=D_{\min}^{\rm NL}$ and $y=\cM_{\rm NL}^{\rm SRE}$,
\begin{equation}
 y\ge3x-\log_2(2C_{\rm sh}).
 \label{sm:eq:SRE-affine}
\end{equation}
The case $x=0$ is immediate. For $0<x\le9/2$ the saturating lower bound of
Eq.~\eqref{eq:SRE-fidelity-main} gives
$y/x\ge-\log_2[(1+2^{-9/2})/2]/(9/2)>1/5$, the ratio being
nonincreasing by concavity, and for $x\ge9/2$
Eq.~\eqref{sm:eq:SRE-affine} gives
$y/x\ge3-\log_2(2C_{\rm sh})/x>1/5$.  This proves the balanced-cut sharpening of the linear bound in Eq.~\eqref{eq:SRE-fidelity-main},
\begin{equation}
 \tfrac15\,D_{\min}^{\rm NL}\le\cM_{\rm NL}^{\rm SRE}\le4\,D_{\min}^{\rm NL}
 \qquad(\text{balanced cuts}).
 \label{sm:eq:SRE-balanced}
\end{equation}

 For the $m$-copy spectrum in the $V>0$ case of
  Theorem~\ref{sm:thm:tensor-dichotomy},
Eq.~\eqref{sm:eq:tensor-counting-bound} implies
$\max_k\widehat w_k=O_\psi(m^{-1/2})$, while the rank window of
Eqs.~\eqref{sm:eq:tensor-rank-upper-tail}
and~\eqref{sm:eq:tensor-rank-lower-tail} meets $O_\psi(\sqrt m)$ shells
and carries a fixed total probability.  Since the shell-max weights
dominate the true shell weights, the power-mean inequality gives
$Q_{\rm sh}=\Theta_\psi(m^{-3/2})$, and Theorem~\ref{sm:thm:shell-sandwich}
yields, for every fixed nonflat spectrum on a balanced cut,
\begin{equation}
 \begin{aligned}
 \cM_{\rm NL}^{\rm SRE}(\psi^{\otimes m})
 &=\frac32\log_2m+O_\psi(1)\\
 &=3D_{\min}^{\rm NL}(\psi^{\otimes m})+O_\psi(1).
 \end{aligned}
 \label{sm:eq:tensor-SRE-law}
\end{equation}

\section{Haar-random entanglement spectra}
\label{sm:Haar}

We use the notation of Methods: $\Psi_{d,D}$ is Haar random on $\mathbb C^d\otimes\mathbb C^D$ with $d=2^\nu\le D$ and $d/D\to c\in(0,1]$, and $\xi_x=d\lambda_x$ are the ordered rescaled eigenvalues of the reduced state. Their distribution converges in probability to the Marchenko--Pastur law
\begin{align}
   \rho_c(x)&=\frac{\sqrt{(b_c-x)(x-a_c)}}{2\pi cx}\,\boldsymbol 1_{[a_c,b_c]}(x),
   \label{sm:eq:MP-density}\\
   a_c&=(1-\sqrt c)^2,\qquad b_c=(1+\sqrt c)^2,\notag
  \end{align}
which has mean $1$ and variance $c$, and $\xi_0\to b_c$~\cite{BaiYin1993,Nadal2011}.
Let $t_c(\alpha)$ be the threshold above which the law $\rho_c$ carries a fraction $\alpha$ of its weight, $\int_{t_c(\alpha)}^{b_c}\rho_c(x)\,dx=\alpha$. The quantities of Methods are then $I_c(\alpha)=\int_{t_c(\alpha)}^{b_c}\sqrt{x}\,\rho_c(x)\,dx$, the contribution of the largest fraction $\alpha$ of the eigenvalues to the mean of $\sqrt x$, and $\Phi_c(\alpha)=I_c(\alpha)^2/\alpha$.

\paragraph{Branch limits.}
For fixed $\alpha=2^{-j}$, the largest $\alpha d$ rescaled eigenvalues fill the interval $[t_c(\alpha),b_c]$ with the density $\rho_c$, and $\xi_0\to b_c$; hence $d^{-1}\sum_{x<\alpha d}\sqrt{\xi_x}\to I_c(\alpha)$ in probability, and the branch of rank $\alpha d$ tends to $\Phi_c(\alpha)$. The maximum over the $\nu+1$ branches can be exchanged with the limit because of the bound
  \begin{equation}
   \max_{j\ge J}\frac{2^j}{d}\Bigl(\sum_{x<d/2^j}\mu_x\Bigr)^2\le2^{-J}d\lambda_0 ,
   \label{sm:eq:Haar-uniform-tail}
  \end{equation}
whose right-hand side tends to $2^{-J}b_c$: for fixed $J$ the finitely many head branches converge jointly, and $J$ is then taken large. Hence $F_{\rm NL}(\ket{\Psi_{d,D}})\xrightarrow{\Pr}F_{\rm H}(c):=\max_{j\ge0}\Phi_c(2^{-j})$.

\paragraph{Only two branches survive.}
 If $X$ has density $\rho_c$ and $E_\alpha=\{X\ge t_c(\alpha)\}$ is the event that $X$ lies among the largest fraction $\alpha$, the Cauchy--Schwarz inequality and $\operatorname{Var}X=c$ give
  \begin{equation}
   \Phi_c(\alpha)\le\mathbb E[X\boldsymbol1_{E_\alpha}]\le\alpha+\sqrt{c\alpha(1-\alpha)},
   \label{sm:eq:MP-covariance-bound}
  \end{equation}
which is at most $(1+\sqrt3)/4$ for $\alpha\le1/4$. The full-rank branch has the closed form
  \begin{equation}
   \Phi_c(1)=\Bigl[{}_2F_1\!\bigl(-\tfrac12,\tfrac12;2;c\bigr)\Bigr]^2,
   \label{sm:eq:full-hypergeometric}
  \end{equation}
whose square root has derivative $-\tfrac18{}_2F_1(\tfrac12,\tfrac32;3;c)<0$, so that $\Phi_c(1)\ge\Phi_1(1)=64/(9\pi^2)>(1+\sqrt3)/4$. Consequently $F_{\rm
 H}(c)=\max\{\Phi_c(1),\Phi_c(1/2)\}$.

\paragraph{Balanced cut.}
At $c=1$, the substitution $x=4\cos^2\varphi$ gives $\alpha=(2\varphi-\sin2\varphi)/\pi$ and $I_1(\alpha)=\tfrac{8}{3\pi}\sin^3\varphi$, so that, with $\varphi_\star$ as in Methods, $\Phi_1(1/2)=\tfrac{128}{9\pi^2}\sin^6\varphi_\star =0.844388\ldots>\Phi_1(1)=\tfrac{64}{9\pi^2}$; this is the $\Delta=0$ case of Eq.~\eqref{eq:Haar-limit-main}.

\paragraph{Imbalanced cuts.}
For $\Delta\ge1$ one has $c=2^{-\Delta}\le1/2$. Euler's integral representation of Eq.~\eqref{sm:eq:full-hypergeometric} reads $\sqrt{\Phi_c(1)}=\mathbb E\sqrt{1-cT}$ with $T\sim\operatorname{Beta}(\tfrac12,\tfrac32)$ and $\mathbb ET=\tfrac14$, so the chord bound for the concave square root gives $\Phi_c(1)\ge[(3+\sqrt{1-c})/4]^2$, whereas Eq.~\eqref{sm:eq:MP-covariance-bound} gives $\Phi_c(1/2)\le(1+\sqrt c)/2$. At $c\le1/2$ the former exceeds the latter, since $[(3+1/\sqrt2)/4]^2-(1+1/\sqrt2)/2=(1-1/\sqrt2)^2/16>0$. Thus the full-rank branch wins, and Eq.~\eqref{sm:eq:full-hypergeometric} is the $\Delta\ge1$ case of Eq.~\eqref{eq:Haar-limit-main}.

Since $0\le F_{\rm NL}\le1$, convergence in probability implies $\mathbb EF_{\rm NL}\to F_{\rm H}(c)$. Moreover, $F_{\rm NL}(\ket\Psi)=\sup_{\phi\in\cA_{\rm NL}}|\braket{\phi|\Psi}|^2$ is a supremum of squared overlaps, each of which changes by at most $2\bigl\|\ket\Psi-\ket{\Psi'}\bigr\|$ between unit vectors, so $|F_{\rm NL}(\ket\Psi)-F_{\rm NL}(\ket{\Psi'})|\le2\bigl\|\ket\Psi-\ket{\Psi'}\bigr\|$. L\'evy's lemma therefore gives $\Pr\bigl(|F_{\rm NL}-\mathbb EF_{\rm NL}|>\varepsilon\bigr)\le2\exp[-\Omega(dD\,\varepsilon^2)]$, and since $d/D\to c>0$ this is $\exp[-\Omega(d^2\varepsilon^2)]$. Finally, $F_{\rm H}(c)>0$, and on the exponentially rare event $F_{\rm NL}<F_{\rm H}(c)/2$ the universal bound $D_{\min}^{\rm NL}\le\log_2(\nu+1)$ applies, while on its complement $-\log_2$ is Lipschitz. Consequently $\mathbb ED_{\min}^{\rm NL}\to-\log_2F_{\rm H}(c)$: typical and mean values have the same limit, as the averages in Fig.~\ref{fig:classes} illustrate.

\section{Critical spectra from the Calabrese--Lefevre ansatz}
\label{sm:CL}

We use the notation of Methods and write $\widehat\Phi_b(w):=G_b(bw^2)^2/N_b(bw^2)$ for the continuous-rank branch at depth $u=bw^2$. Uniformly on compact subsets of $0<w<2$,
\begin{align}
   N_b(bw^2)
   &=\frac{e^{2bw}}{2\sqrt{\pi b}\,w^{1/2}}\,[1+O(b^{-1})],
   \label{sm:eq:CL-count-asymptotic}\\
   G_b(bw^2)
   &=\frac{e^{b(-1/2-w^2/2+2w)}}{\sqrt{\pi b}\,w^{1/2}(2-w)}\,[1+O(b^{-1})].
   \label{sm:eq:CL-amplitude-asymptotic}
\end{align}
The first line is the Bessel asymptotic $I_\nu(z)=e^z(2\pi z)^{-1/2}[1+O(z^{-1})]$ applied to $N_b=I_0(2bw)$. For the second, the substitution $v=bs^2$ turns the integral defining $G_b$ into $\sqrt{b/\pi}\int_0^w s^{-1/2}e^{b(2s-s^2/2)}[1+O(b^{-1})]\,ds$; since the exponent increases for $s<2$, the integral is dominated by its upper endpoint, $\int_0^w e^{b\varphi(s)}g(s)\,ds=e^{b\varphi(w)}g(w)/[b\varphi'(w)]\,[1+O(b^{-1})]$ with $\varphi'(w)=2-w$, which produces the factor $(2-w)^{-1}$ and restricts the expansion to $w<2$. Squaring the second line and dividing by the first gives Eq.~\eqref{eq:CL-envelope-methods}.

To control the tails, for $w\le\tfrac14$ the bound  $\widehat\Phi_b(w)\le e^{-b}N_b(bw^2)$ together with $I_0(x)\le e^x$, and for $w\ge\tfrac74$ the exact total amplitude $G_b(\infty)=R_{1/2}=e^{3b/2}$ together with $I_0(x)\ge e^x/(2\sqrt{2\pi x})$ for $x\ge1$ (from the integral representation of $I_0$)d, give
  \begin{equation}
   \begin{aligned}
   \widehat\Phi_b(w)
   &\le e^{-b+2bw}\le e^{-b/2},
   & w&\le\tfrac14,\\
   \widehat\Phi_b(w)
   &\le C\sqrt{bw}\,e^{b(3-2w)}
   \le C'\sqrt b\,e^{-b/2},
   & w&\ge\tfrac74,
   \end{aligned}
   \label{sm:eq:CL-tails}
  \end{equation}
where $C,C'>0$ are independent of $b$ and $w$ for all sufficiently  large $b$. Both bounds are exponentially smaller than the saddle value $2/\sqrt{\pi b}$.

On the remaining compact interval, the Gaussian factor in  Eq.~\eqref{eq:CL-envelope-methods} exponentially suppresses the complement of every fixed neighborhood of $w=1$. A differentiable uniform expansion gives $\partial_w\ln\widehat\Phi_b(w) =-2b(w-1)-1/(2w)+2/(2-w)+O(b^{-1})$ and $\partial_w^2\ln\widehat\Phi_b(w)=-2b+O(1)$. Hence $w_\star=1+3/(4b)+O(b^{-2})$, with the value quoted in Methods. Since $\frac{d}{dw}\ln N_b(bw^2)=2b+O(1)$ near $w=1$, shifting $\ln r$ by at most $\tfrac12\ln2$ moves $w$ by $O(b^{-1})$. The first derivative vanishes at the saddle and the logarithmic curvature is $O(b)$, so this changes the value only by a relative $O(b^{-1})$. Finally, $u_\star=b+O(1)$.

The strict Bessel spectrum assumes that the nonuniversal prefactor of  the CFT moments is independent of the R\'enyi index~\cite{Calabrese2008spectrum}. The leading result persists for
\begin{equation}
   R_\alpha=f(\alpha)\,e^{-b(\alpha-1/\alpha)}[1+o(1)],
   \label{sm:eq:CL-prefactor}
\end{equation}
provided the saddle inversion and its differentiated remainder hold uniformly in a complex neighborhood of the relevant saddles, with $f$ analytic and nonzero there, $f(1)=1$, and all branches outside a fixed neighborhood of $w=1$ remaining negligible. Under these assumptions, saddle inversion multiplies the envelope by $f(1/w)$. Its value at $w=1$ is fixed by normalization, while its derivatives shift $\ln r_\star$ only by $O(1)$. Smoothness on the real axis alone, or fixed-$\alpha$ CFT scaling without such control, is not sufficient.

\section{Entanglement embezzlement: finite-error bounds and examples}
\label{sm:embezzlement}

Here, we relate $F_{\rm NL}$ to finite LOCC conversion errors, evaluate the canonical
harmonic catalysts, and separate LOCC from local-unitary embezzlement.

\subsection{Finite-error operational bounds}

Let $T_{\rm LOCC}(\rho\to\sigma):=\inf_{\Lambda\in{\rm LOCC}} \tfrac12\|\Lambda(\rho)-\sigma\|_1$.  For pure states, the trace-star distance $T_\star$ is the minimum total-variation distance between the target Schmidt vector and an output Schmidt vector that majorizes the source vector~\cite{Nielsen1999Majorization}, and Ref.~\cite{ZanoniTheurerGour2024} proved $\tfrac12T_\star^2\le T_{\rm LOCC}\le\sqrt{2T_\star}$ together with an exact partial-sum formula for $T_\star(\chi\to\chi\otimes\Phi_1)$~\cite{ZanoniTheurerGour2024}.  For a catalyst $\chi$ with ordered Schmidt probabilities, every term of that formula is at most $\eta_{\rm emb}$, while the term associated with a maximizing octave is at least $\eta_{\rm emb}/2$.  Therefore, with $C=3+2\sqrt2$,
\begin{equation}
 \frac{F_{\rm NL}}{2C}
 \le T_\star(\chi\to\chi\otimes\Phi_1)
 \le F_{\rm NL},
 \label{sm:eq:fidelity-star}
\end{equation}
and
\begin{equation}
 \frac{F_{\rm NL}^2}{8C^2}
 \le T_{\rm LOCC}(\chi\to\chi\otimes\Phi_1)
 \le\min\{1,\sqrt{2F_{\rm NL}}\}.
 \label{sm:eq:fidelity-trace}
\end{equation}
For a target $\sigma$ preparable from the rank-$m$ maximally entangled
state, in particular any target supported on local dimensions at most
$m$, the star-distance triangle inequality gives
$T_{\rm LOCC}(\chi\to\chi\otimes\sigma)
\le\min\{1,\sqrt{2\lceil\log_2m\rceil\,F_{\rm NL}}\}$.  The nonlocal stabilizer fidelity is thus constant-factor equivalent to the trace-star defect, but only polynomially related to the physical trace error.

\subsection{Canonical harmonic catalysts}

For the harmonic spectrum $\lambda_x=[(x+1)H_d]^{-1}$, $0\le x<d$, the weight of the octave $[\ell,2\ell]$ is $(H_{2\ell+1}-H_\ell)/H_d$, which is largest at $\ell=0$; hence $\eta_{\rm
  emb}(\Gamma_d)=1/H_d$ exactly. The branch of rank $q$ is $(qH_d)^{-1}\bigl(\sum_{i=1}^{\min\{q,d\}}i^{-1/2}\bigr)^2$, and the integral estimates
  \begin{equation}
   2\bigl(\sqrt{q+1}-1\bigr)\le\sum_{i=1}^{q}i^{-1/2}<2\sqrt q
   \label{sm:eq:vDH-sum-bounds}
  \end{equation}
  show that every branch, including the one at the padded rank $2^{\lceil\log_2d\rceil}$, lies below $4/H_d$, while the branch at the largest dyadic rank $q\le d$ tends to $4/H_d$. Hence
  $F_{\rm NL}(\Gamma_d)=4[1+o(1)]/H_d$, i.e. $D_{\min}^{\rm NL}(\Gamma_d)=\log_2\ln d-2+o(1)$, for every $d$ and not only along a dyadic subsequence. The dyadic shells carry $p_0=1/H_d$
  and $p_k=(H_{2^k}-H_{2^{k-1}})/H_d=[\ln2+O(2^{-k})]/H_d$ for $1\le k\le\lfloor\log_2d\rfloor$: the canonical catalyst spreads its Schmidt weight almost uniformly over logarithmic rank
  scales, which is why it lies within $O(1)$ of the entropy bound in Fig.~\ref{fig:classes}. It is one explicit universal embezzling family, not a characterization of all of
  them~\cite{LeungWang2014embezzling}.

\subsection{LOCC versus local-unitary embezzlement}

We construct a family that is universal under LOCC but fails to embezzle fails to embezzle one Bell pair with the local-unitary protocol of Eq.~\eqref{eq:embezzlement-protocol}.
Fix $0<a<1/3$, take $n\ge3$, and assign to the dyadic shells $D_0=\{0\}$ and $D_k=\{2^{k-1},\ldots,2^k-1\}$, $k=1,\ldots,n-1$, the probabilities
\begin{equation}
 \begin{aligned}
 \lambda_x^{(n)}&=\frac{w_k}{|D_k|}\quad(x\in D_k),\\
 w_k&=\frac{1+a(-1)^k}{Z_n},\qquad
 Z_n=\sum_{k=0}^{n-1}[1+a(-1)^k],
 \end{aligned}
 \label{sm:eq:alternating-spectrum}
\end{equation}
with $Z_n=n$ for even $n$ and $Z_n=n+a$ for odd $n$.  The spectrum is decreasing because the only adverse boundaries, from odd to even $k$, obey $(1+a)/[2(1-a)]<1$, and every octave meets at most two adjacent shells, so $(1+a)/Z_n\le\eta_{\rm emb}\le2/Z_n$: the family is universal under LOCC by the criterion of Ref.~\cite{ZanoniTheurerGour2024} and has $F_{\rm NL}=\Theta(n^{-1})$.  
For the local-unitary protocol, Lemma~\ref{lem:Schmidt-overlap} gives the optimal root fidelity for extracting one Bell pair as
\begin{equation}
 \begin{aligned}
 \mathfrak f_n^{\rm loc}
 &=\sum_{x=0}^{2^{n-1}-1}
 \sqrt{\frac{\lambda_x^{(n)}
 \lambda_{\lfloor x/2\rfloor}^{(n)}}{2}}\\
 &=\frac{(1+a)/\sqrt2+\sqrt{(1-a^2)/2}+(n-2)\sqrt{1-a^2}}{Z_n},
 \end{aligned}
 \label{sm:eq:local-isometry-overlap}
\end{equation}
which tends to $\sqrt{1-a^2}<1$: the trace-distance error tends to $a$. Thus $F_{\rm NL}\to0$ characterizes LOCC embezzlement but does not supply the scale covariance required by the local-unitary protocol.  This gives an elementary explicit finite-dimensional instance of the separation discussed in Ref.~\cite{Luijk2025multipartite}.

\section{Prime-dimensional qudits}
\label{sm:qudits}
We fix the onsite dimension $p$ to be prime number, $\nu=\min(n_A,n_B)$, and write the decreasing, zero-padded Schmidt amplitudes as $\mu_x$, $0\le x<p^\nu$. 
For $x=\sum_{\ell=0}^{\nu-1}x_\ell p^\ell$, let $\ket{[x]_p}_P$ denote the digit string $(x_0,\ldots,x_{\nu-1})\in\mathbb F_p^{\nu}$ encoded on $\nu$ designated qudits of party $P=A,B$, with $\ket0$ on every remaining qudit.
Let $\cC_N^{(p)}$ denote the prime-qudit Clifford group, whose orbit $\cC_N^{(p)}\ket{0^{\otimes N} }$ is the set of pure prime-qudit stabilizer states, and define the nonlocal stabilizer fidelity for qudits as
  \begin{equation}
   F_{\rm NL}^{(p)}( \ket{\Psi} ):=
   \max_{U_A,U_B}\max_{C\in\cC_N^{(p)}}
   \left|\bra{0^{\otimes N}}C^\dagger(U_A\otimes U_B)\ket\Psi\right|^2.
   \label{sm:eq:prime-fidelity-definition}
  \end{equation}

  \begin{theorem}[Prime-qudit formula]
  \label{sm:thm:prime-qudit}
  \begin{equation}
   F_{\rm NL}^{(p)}(\ket{\Psi})
   =\max_{0\le k\le\nu}
   \frac1{p^k}\left(\sum_{x=0}^{p^k-1}\mu_x\right)^2,
   \label{sm:eq:prime-formula}
  \end{equation}
  and the sorted base-$p$ representative $\ket{\Psi_{\rm CB}^{(p)}}:=\sum_{x=0}^{p^\nu-1}\mu_x\ket{[x]_p}_A\ket{[x]_p}_B$ maximizes the stabilizer fidelity over all local unitaries
  $U_A,U_B$.
  \end{theorem}

\begin{proof}
The proof of Theorem~\ref{thm:main} in Methods uses only two stabilizer-specific inputs: stabilizer Schmidt spectra are flat of dyadic rank, and every allowed rank is realized by a Bell-pair state. Their prime-$p$ analogues hold, and the remaining steps carry over verbatim with $2\to p$.

The first analogue follows from the prime-qudit normal  form~\cite{LooiGriffiths2011}: across any cut, a stabilizer state is equivalent under party-local Clifford operations to $k$ generalized Bell pairs and $\ket0$ spectators, for some $0\le k\le\nu$. Its Schmidt spectrum is therefore flat of rank $p^k$. Maximizing the fixed-rank overlap supplied by Lemma~\ref{lem:Schmidt-overlap} over all stabilizer states gives the upper bound in Eq.~\eqref{sm:eq:prime-formula}.

For the second analogue, $\ket{\Phi_{p,k}}:=p^{-k/2}\sum_{x=0}^{p^k-1} \ket{[x]_p}_A\ket{[x]_p}_B$ consists of $k$ generalized Bell pairs on the designated qudits carrying the $k$ least significant digits and $\ket0$ elsewhere. It is therefore a stabilizer state of Schmidt rank $p^k$ for every $0\le k\le\nu$. As in Methods, choose an ordered Schmidt decomposition padded with zero amplitudes and complete the associated Schmidt vectors where necessary. Local unitaries map these vectors to the labels $\ket{[x]_p}_A,\ket{[x]_p}_B$ and are extended arbitrarily on the larger party. Thus $\ket{\Psi_{\rm CB}^{(p)}}$ lies in the local-unitary orbit of $\ket\Psi$, and its squared overlap with $\ket{\Phi_{p,k}}$ is the $k$th term of Eq.~\eqref{sm:eq:prime-formula}. The upper and lower bounds coincide, proving Eq.~\eqref{sm:eq:prime-formula} and showing that $\ket{\Psi_{\rm CB}^{(p)}}$ simultaneously saturates every Schmidt-rank sector.
  \end{proof}

\end{document}